\documentclass[final,twoside]{IEEEtran} 

\usepackage{float}
\usepackage{amssymb,amsmath,amsthm,bm}

\newcommand{\Mod}[1]{\ (\mathrm{mod}\ #1)}
\makeatletter
\newcommand{\tpmod}[1]{{\@displayfalse\pmod{#1}}}
\makeatother

\usepackage{graphicx,color}
\graphicspath{{figures-pdf/}}

\usepackage{cite}
\usepackage{url}
\usepackage{diagbox}
\usepackage[aboveskip=1pt]{subcaption}

\usepackage{algorithm}
\usepackage{algpseudocode}
\usepackage{multirow,bigstrut}
\usepackage{tabularx}
\usepackage{arydshln}
\usepackage{empheq}
\usepackage{datetime}
\newlength\OneImW
\newlength\BigOneImW
\newlength\twofigwidth
\newlength\ThreeImW
\newlength\FourImW
\newlength\DoubleThreeImW
\newlength\sfigwidth
\newlength\vfigskip
\newlength\figsep
\newtheorem{Proposition}{Proposition}
\newtheorem{Property}{Property}
\newtheorem{Corollary}{Corollary}

\newtheorem{theorem}{Theorem}

\newtheorem{lemma}{Lemma}

\DeclareMathOperator*{\lcm}{lcm}

\newcommand{\mtx}[1]{\bm{#1}}

\newcommand{\Fee}{\mtx{\Phi}}
\usepackage[bookmarks=false]{hyperref}
\hypersetup{
 linktocpage=true, pdfborderstyle={/S/S/W 1}, hyperindex=true, bookmarks=true, bookmarksopen=true, bookmarksnumbered=true,
}

\begin{document}

\title{The Graph Structure of the Generalized Discrete Arnold's Cat Map}

\author{Chengqing Li, Kai Tan, Bingbing Feng, Jinhu L\"u
\thanks{This work was supported by the National Natural Science Foundation of China (no.~61772447, 61532020).}

\thanks{C. Li is with College of Computer Science and Electronic Engineering, Hunan University, Changsha 410082, Hunan, China (DrChengqingLi@gmail.com).}

\thanks{K. Tan and B. Feng are with College of Information Engineering, Xiangtan University, Xiangtan 411105, Hunan, China.}

\thanks{J. L\"u is with School of Automation Science and Electrical Engineering, Beihang University, Beijing 100083, China}
}

\markboth{IEEE Transactions}{Li \MakeLowercase{\textit{et al.}}}

\IEEEpubid{\begin{minipage}{\textwidth}\ \\[12pt] \centering
1549-8328 \copyright 2019  IEEE. Personal use is permitted, but republication/redistribution requires IEEE permission.\\
  See http://www.ieee.org/publications\_standards/publications/rights/index.html for more information.
   \\ \scriptsize \today{} (\currenttime)
\end{minipage}}

\maketitle

\begin{abstract}
Chaotic dynamics is an important source for generating pseudorandom binary sequences (PRNS).
Much efforts have been devoted to obtaining period distribution of the generalized discrete Arnold's Cat map
in various domains using all kinds of theoretical methods, including
Hensel's lifting approach. Diagonalizing the transform matrix of the map,
this paper gives the explicit formulation of any iteration of the generalized Cat map. Then, its real graph (cycle) structure in any binary arithmetic domain is disclosed. The subtle rules on how the cycles (itself and its distribution) change with the arithmetic precision $e$ are elaborately investigated and proved. The regular and beautiful patterns of Cat map demonstrated in a computer adopting fixed-point arithmetics are rigorously proved and experimentally verified. The results will facilitate research on dynamics of variants of the Cap map in any domain and its effective application in cryptography.
In addition, the used methodology can be used to evaluate randomness of PRNS generated by
iterating any other maps.
\end{abstract}

\section{Introduction}

\IEEEPARstart{P}{eriod} and cycle distribution of chaotic systems are fundamental
characteristics measuring their dynamics and function, and supporting their practical values \cite{Percival:CAT:PD87,shi:homo:NSR19}. As the most popular application form, various digitized chaotic systems were constructed or enhanced as a source of producing random number sequences: Tent map \cite{PAPADOPOULOS:Tent:TIT1995},
Logistic map \cite{Chen:Logistic:TCASII10,garcia2018chaos:TIM18}, Cat map \cite{Hua:cat:TC2018}, Chebyshev map of even degree \cite{Kohda:ITIT:1997}, piecewise linear map \cite{Umeno:ITIT:2013},
and Chua's attractor \cite{cqli:Diode:TCASI19}. Among them,
Arnold's Cat map
\begin{equation}
f(x, y)=(x+y, x+2y) \bmod 1
\label{eq:oriArnold}
\end{equation}
is one of the most famous chaotic maps, named after Vladimir Arnold, who
heuristically demonstrated its stretching (mixing) effects using an image of a Cat in \cite[Fig. 1.17]{arnol1968mathematical}. Attracted by the simple form but complex dynamics of Arnold's Cat map, it is adopted as a hot research object in various domains: quadratic field
\cite{Percival:CAT:PD87}, two-dimensional torus \cite{Barash:CAT:PRE2006,Isaeva:cat:PRE2006,Ermann:Cat:PDP2012}, \cite{Okayasu:PAMS:2010}, finite-precision digital computer \cite{chen2012periodpe,Catchen2013period2e}, quantum computer \cite{Kurlberg:cat:AM2005,Horvat:cat:JPAT2007,Moudgalya:PR:2019}. In \cite{Penrose:entropy:PTRSAPES13}, Cat map is used as an example to define a microscopic entropy of chaotic systems. The nice properties of Cat map demonstrated in the infinite-precision domains seemingly support that it is widely used in many cryptographic applications, e.g.
chaotic cryptography \cite{Farajallah:IJBC:2016}, image encryption \cite{Chenlei:CBM:2015}, image privacy protection \cite{cqli:IEAIE:IE18,cqli:meet:JISA19}, hashing scheme \cite{Kanso:hash:ND2015}, PRNG \cite{Falcioni:PRNS:PRE2005,Barash:CAT:PRE2006}, random perturbation \cite{Yarmola:cat:ETDS2011}, designing unpredictable path flying robot \cite{Curiac:path:DSJ2015}.

Recently, the dynamics and randomness of digital chaos are investigated from the perspective of functional graphs (state-mapping networks) \cite{cqli:autoblock:IEEEM18}. In \cite{cqli:network:TCASI2019}, how the structures of Logistic map and Tent map change with the implementation precision $e$ is theoretically proved. Some properties on period of a variant of Logistic map over Galois ring $\mathbb{Z}_{3^{e}}$ are presented \cite{Yang:Logsitic:SP2018,Li:IJPRAI:2019}. In \cite{Frahm:PR:2018},
the phase space of cat map is divided into some uniform Ulam cells, and the associated directed complex network is built with respect to mapping relationship between every pair of cells. Then, the average path length of the network is used to measure the underlying dynamics of Cat map. In \cite{Victor:detect:RCD18}, the elements measuring phase space structures of Cat map, fixed points, periodic orbits and manifolds (stable or unstable), are detected with Lagrangian descriptors. In \cite{Daniel:graph:DCC2019}, the functional graph of general linear maps over finite fields is studied with various network parameters, e.g.
the number of cycles and the average of the pre-period (transient) length. To quickly calculate the maximal transient length, fixed points
and periodic limit cycles of the functional graph of digital chaotic maps, a fast period search algorithm using a tree structure is designed in \cite{Ding:period:ND19}.

\IEEEpubidadjcol 

The original Cat map~(\ref{eq:oriArnold}) can be attributed to the general matrix form
\begin{equation}
f(\textbf{x})=(\Fee \cdot \textbf{x})\bmod N,
\label{eq:ArnoldMatrix}
\end{equation}
where $N$ is a positive integer, $\textbf{x}$ is a vector of size $n\times 1$, and $\Fee$ is a matrix of size $n\times n$.
The determinant of the transform matrix $\Fee$ in Eq.~(\ref{eq:ArnoldMatrix}) is one, so the original Cat map is area-preserving.
Keeping such fundamental characteristic of Arnold's Cat map unchanged, it can be generalized or extended via various strategies:
changing the scope (domain) of the elements in $\Fee$ \cite{Chen:CSF:2004}; extending the transform matrix to 2-D, 3-D and even any higher dimension \cite{kwok2007generalCatMap,Hua:cat:TC2018}; modifying the modulo $N$ \cite{Wu:catmapITC:2016}; altering the domain of some parameters or variables \cite{Sano:sawtooth:PRE2002}.

Among all kinds of generalizations of Cat map, the one in 2-D integer domain
received most intensive attentions due to its direct application on permuting position of elements of image data, which can be represented as
\begin{equation}
f
\begin{bmatrix}
		x_{n} \\
		y_{n}
	\end{bmatrix}=
	\begin{bmatrix}
		x_{n+1} \\
		y_{n+1}
	\end{bmatrix}
	=
	\textbf{C}\cdot
	\begin{bmatrix}
	x_{n} \\
	y_{n}
	\end{bmatrix}\bmod N,
	\label{eq:ArnoldInteger}
\end{equation}
where
\begin{equation}
\textbf{C}=
\begin{bmatrix}
1 & p    \\
q & 1+p\cdot q
\end{bmatrix},
\label{eq:MatMatrix}
\end{equation}
$x_n$, $y_n\in \mathbb{Z}_N$, and $p, q, N\in \mathbb{Z}^+$.

In this paper, we refer to the generalized Cat map~(\ref{eq:ArnoldInteger})
as \textit{Cat map} for simplicity. In \cite{dyson1992periodCat}, the upper and lower bounds of the period of Cat map~(\ref{eq:ArnoldInteger}) with $(p, q)=(1, 1)$
are theoretically derived. In \cite{Bao:ND:2012}, the corresponding properties of Cat map (\ref{eq:ArnoldInteger}) with $(p, q, N)$ satisfying some constrains
are further disclosed. In \cite{chen2012periodpe,Chen:linearCNSNS:2012,Catchen2013period2e,chen2014period:TCS14}, F. Chen systematically analyzed the precise period distribution of Cat map (\ref{eq:ArnoldInteger}) with any parameters. The whole analyses are divided into three parts according to influences on algebraic properties of $(\mathbb{Z}_{N}, +, \cdot)$ imposed by $N$: a Galois field when $N$ is a prime \cite{Chen:linearCNSNS:2012}; a Galois ring when $N$ is a power of a prime \cite{chen2012periodpe,Catchen2013period2e}; a commutative ring when $N$ is a common composite \cite{chen2014period:TCS14}.
According to the analysis methods adopted, the second case is further divided into two sub-cases $N=p^e$ and $N=2^e$, where $p$ is a prime larger than or equal to 3 and $e$ is an integer. From the viewpoint of real applications in digital devices, the case of Galois ring $\mathbb{Z}_{2^e}$ is of most importance since it is isomorphic to the set of numbers represented by $e$-bit fixed-point arithmetic format with
operations defined in the standard for arithmetic of computer.

The period of a map over a given domain is the least common multiple
of the periods of all points in the domain. Confusing periods of two different objects causes some misunderstanding on impact of the knowledge about period distribution of Cat map in some references like \cite{Bao:ND:2012}. What's worse, the local
properties of Arnold's Cat map are omitted.
Diagonalizing the transform matrix of Cat map with its eigenmatrix, this paper derives the explicit representation of any iteration of Cat map. Then, the evolution properties of the internal structure
of Cat map (\ref{eq:ArnoldInteger}) with incremental increase of $e$
are rigorously proved, accompanying by some convincing experimental results.

The rest of this paper is organized as follows. Section
~\ref{sec:Previous} gives previous works on deriving the period distribution of Cat map. Section~\ref{sec:main} presents some properties on structure of Cat map. Application of the obtained results are discussed in
Sec.~\ref{sec:apply}. The last section concludes the paper.

\newcommand{\minitab}[2][l]{\begin{tabular}{#1}#2\end{tabular}}

\setlength\tabcolsep{2pt} 
\addtolength{\abovecaptionskip}{0pt}
\renewcommand{\arraystretch}{1.2}
\begin{table*}[!htb]
	\caption{The conditions of $(p, q)$ and the number of their possible cases, $N'_T$, corresponding to a given $T$.}
	\centering 
	\begin{tabular}{*{3}{c|}c} 
		\hline 
		$T$  & $p$ & $q$ & $N'_T$ \\ \hline
		1    & 0   &  0  &  1     \\ \hline
		\multirow{2}{*}{2}   & $p\bmod 2^{e}=2^{e-1}$  & $q\bmod 2^{e-1}=0$    &   2   \\ \cline{2-4}
		& $q\bmod 2^{e}=0$      & $p\bmod 2^{e}=2^{e-1}$ &    1      \\\hline
		3  &  $p\equiv 1 \bmod 2$    &  $q\equiv p^{-1}(2^e-3)\bmod 2^e$   &  $2^{e-1}$   \\ \hline
		\multirow{6}{*}{4}   & $p\equiv 1 \bmod 2$        &  $q\equiv p^{-1}(2^e-2)\bmod 2^e$              &   $2^{e-1}$            \\ \cline{2-4}   	
		& $p\equiv 0 \bmod 2$,  $p\not\equiv 0 \bmod 4$   &  $q\equiv (p/2)^{-1}(2^{e-1}-1)\bmod 2^{e-1}$   &  $2^{e-1}$  \\  \cline{2-4}   	
		&  $p\equiv 1 \bmod 2$    & $q\equiv p^{-1}(2^{e-1}-2)\bmod 2^e$  &    $2^{e-1}$   \\ \cline{2-4}   	
		&  $p\equiv 0 \bmod 2$    &  $q\equiv (p/2)^{-1}(2^{e-1}-2)\bmod 2^e$  &    $2^{e-1}$  \\ \cline{2-4}
		& $p\bmod 2^{e-1}=2^{e-2}$  & $q\bmod 2^{e-2}=0$ &  8     \\ \cline{2-4}
		& $q\bmod 2^{e-1}=0$  & $p\bmod 2^{e-1}=2^{e-2}$ &  4      \\ \hline
		\multirow{3}{*}{6}    &  \multirow{3}{*}{$p\equiv 1 \bmod 2$}        &  $q\equiv p^{-1}(2^{e-1}-3)\bmod 2^e$ & $2^{e-1}$   \\  	 \cline{3-4}
		&  & $q\equiv p^{-1}(2^{e-1}-1)\bmod 2^e$       & $2^{e-1}$   \\\cline{3-4}   	
		& & $q\equiv p^{-1}(2^{e}-1) \bmod 2^e$    & $2^{e-1}$   \\\hline
		\multirow{4}{*}{\minitab[c]{$2^k$,\\ $k\in\{3, 4, \ldots, e-1\}$}}      	
		& $p\equiv 1 \bmod 2$           &  $q\equiv p^{-1}(2^{e-k+1}l-2)\bmod 2^e$, $l\equiv 1 \bmod 2$, $l\in [1, 2^{k-1}-1]$    & $2^{e+k-3}$    \\   \cline{2-4}
		& $p\equiv 0 \bmod 2$        &  $q\equiv (p/2)^{-1}(2^{e-k+1}l-2)\bmod 2^e$, $l\equiv 1 \bmod 2$, $l\in [1, 2^{k-1}-1]$   & $2^{e+k-3}$  \\  \cline{2-4}
		& $p\bmod 2^{e-k+1}=2^{e-k}$ & $q\bmod 2^{e-k}=0$ & $2^{2k-1}$  \\\cline{2-4}
		& $p\bmod 2^{e-k+1}=0$ & $q\bmod 2^{e-k+1}=2^{e-k}$ & $2^{2k-2}$ \\  \hline
		\multirow{2}{*}{$2^e$} & $p\equiv 1 \bmod 2$       &   $q\equiv 0 \bmod 4$    &  $2^{2e-3}$  \\  \cline{2-4}
		& $p\equiv 0 \bmod 4$       &   $q\equiv 1 \bmod 2$    &  $2^{2e-3}$    \\  \hline
		\multirow{2}{*}{\minitab[c]{$3\cdot 2^k$,\\ $k\in\{2, 3, \ldots, e-2\}$}}
		& \multirow{2}{*}{$p\equiv 1 \bmod 2$}     &  $q\equiv p^{-1}(2^{e-k}l-3)\bmod 2^e$, $l\equiv 1 \bmod 2$, $l\in [1, 2^{k}]$    &  $2^{e+k-2}$   \\  \cline{3-4}
		&    & $q\equiv p^{-1}(2^{e-k+1}l-1)\bmod 2^e$, $l\equiv 1 \bmod 2$, $l\in [1, 2^{k-1}-1]$   & $2^{e+k-2}$     \\ \hline 	
	\end{tabular}
	\label{table:num2} 
\end{table*}

\section{The previous works on the period of Cat map}
\label{sec:Previous}

To make the analysis on the (overall and local) structure of Cat map complete, the previous related elegant results are briefly reviewed in this section.

When $(p, q)=(1, 1)$,
$
\textbf{C}=
\begin{bsmallmatrix}
1   & 1   \\
1   & 2
\end{bsmallmatrix}=
\begin{bsmallmatrix}
0   & 1   \\
1   & 1
\end{bsmallmatrix}^2
$, the period problem of Cat map~(\ref{eq:ArnoldInteger})
can be transformed as the divisibility properties of Fibonacci numbers \cite{dyson1992periodCat}. Then, the known theorems about Fibonacci numbers are used to
obtain the upper and lower bounds of the period of Cat map~(\ref{eq:ArnoldInteger}), i.e.
\begin{equation}
\log_{\lambda_+}(N)<T\le 3N,
\label{eq:p1q1Bound}
\end{equation}
where $\lambda_+=(1+\sqrt{5})/2$. Under specific conditions on prime decomposition forms of $N$ or the parity of $T$, the two bounds in Eq.~(\ref{eq:p1q1Bound}) are further optimized
in \cite{dyson1992periodCat}.

When $N$ is a power of two, as for any $(p, q)$, \cite{Catchen2013period2e} gives possible representation form of the period of Cat map~(\ref{eq:ArnoldInteger}) over Galois ring $\mathbb{Z}_{2^e}$, shown in Property~\ref{prop:period}. Furthermore, the relationship between $T$ and the number of
different Cat maps possessing the period, $N_T$, is precisely derived:
\begin{IEEEeqnarray}{rCl}
	\IEEEeqnarraymulticol{3}{l}{N_T=}   \nonumber \\
	\begin{cases}
		1                         & \mbox{if } T=1;  \\
		3                         & \mbox{if } T=2;  \\	
		2^{e+1}+12                & \mbox{if } T=4;  \\
		2^{e-1}+2^{e}             & \mbox{if } T=6;  \\
		2^{e+k-2}+3\cdot 2^{2k-2} & \mbox{if } T=2^k,\ k\in \{3, 4, \cdots, e-1\};\\
		2^{2e-2}                  & \mbox{if } T=2^e; \\
		2^{e+k-1}                 & \mbox{if } T=3\cdot 2^k, k\in \{0, 2, 3, \cdots, e-2\},
	\end{cases}
	\label{eq:numberMaps}
\end{IEEEeqnarray}
where $e\ge 4$.

\begin{Property}
	The representation form of $T$ is determined by parity of $p$ and $q$:
	\begin{equation}
	T=
	\left.\begin{cases}
	2^k,           & \mbox{if }  2\mid p \mbox{ or } 2\mid q;\\
	3\cdot 2^{k'}, & \mbox{if }  2\nmid p \mbox{ and } 2\nmid q.
	\end{cases}\right.
	\end{equation}
	where $k\in \{0, 1, \cdots, e\}$, $k'\in \{0, 1, \cdots, e-2\}$.
	\label{prop:period}
\end{Property}

When $e=3$,
\begin{IEEEeqnarray*}{rCl}
N_T=\begin{cases}
	1              & \mbox{if } T=1;  \\
	3              & \mbox{if } T=2;  \\	
    2^{e-1}        & \mbox{if } T=3;  \\	
	2^{e+1}+12     & \mbox{if } T=4;  \\
	2^{e-1}+2^{e}  & \mbox{if } T=6;  \\		
	2^{2e-2}       & \mbox{if } T=8,
\end{cases}
\end{IEEEeqnarray*}
which cannot be presented as the general form (\ref{eq:numberMaps}) as \cite[Table III]{Catchen2013period2e}, e.g. $2^{2e-2}\neq 2^{e+k-2}+3\cdot 2^{2k-2}$ when $e=k=3$. From Eq.~(\ref{eq:numberMaps}), one can see that there are $(1+3+2^{e-1}+2^{e+1}+12+2^{e-1}+2^{e})=2^{e+2}+16$ generalized Arnold's maps
whose periods are not larger than 6, which is a huge number for ordinary digital computer, where $e\ge 32$.

The generating function of the sequence generated by iterating Cat map~(\ref{eq:ArnoldInteger}) over $\mathbb{Z}_{2^e}$ from initial point $(x_0, y_0)$ can be represented as
\begin{equation*}
X(t)=\frac{g_x(t)}{f(t)}
\end{equation*}
and
\begin{equation*}
Y(t)=\frac{g_y(t)}{f(t)},
\end{equation*}
where
\begin{equation*}
\begin{bmatrix}
g_x(t) \\
g_y(t)
\end{bmatrix}=
\begin{bmatrix}
-1-p\cdot q  & p    \\
q            & -1
\end{bmatrix}\cdot
\begin{bmatrix}
x_0 \\
y_0
\end{bmatrix}\cdot t+
\begin{bmatrix}
x_0 \\
y_0
\end{bmatrix},
\end{equation*}
and
\begin{equation*}
f(t)=t^2-((pq+2) \bmod 2^e)\cdot t+1.
\end{equation*}
Referring to Property~\ref{prop:multiplecycle}, when $p$ and $q$ are not both even, the period of Cat map is equal to the period of $f(t)$. So the period problem of Arnold's Cat map becomes that of a decomposition part of its generation function. First, the number of distinct Cat maps possessing a specific period over
$\mathbb{Z}_{2}[t]$ is counted. Then, the analysis is incrementally extended to $\mathbb{Z}_{2^e}[t]$ using the Hensel's lifting approach.
As for any given value of the period of Arnold's Cat map, all possible values of the corresponding $(p, q)$ are listed in Table~\ref{table:num2} \footnote{
To facilitate reference of readers, we re-summarized the results in \cite{Catchen2013period2e} in a concise and straightforward form.}.

\begin{Property}
As for Cat map~\eqref{eq:ArnoldInteger} implemented over $(\mathbb{Z}_{2^{e}}, +, \cdot)$,
there is one point in the domain, whose period is a multiple of the period of any other points.
\label{prop:multiplecycle}	
\end{Property}

\section{The structure of Cat map over $(\mathbb{Z}_{2^e}, +,\ \cdot\ )$}
\label{sec:main}

First, some intuitive properties of Cat map over $(\mathbb{Z}_{2^e}, +,\ \cdot\ )$ are presented. Then, some general properties of Cat map over
$(\mathbb{Z}_{N}, +,\ \cdot\ )$ and $(\mathbb{Z}_{2^e}, +,\ \cdot\ )$ are given,
respectively. Finally, the regular graph structures of Cat map over $(\mathbb{Z}_{2^e}, +,\ \cdot\ )$ are disclosed
with the properties of two parameters of Cat map's explicit presentation matrix.

\subsection{Properties of funtional graph of Cat map over $(\mathbb{Z}_{2^e}, +,\ \cdot\ )$}
\label{ssec:smn}

\setlength{\textfloatsep}{5pt}
\setlength{\abovecaptionskip}{3pt}

\begin{figure*}[!htb]
	\centering
	\begin{minipage}[t]{0.4\twofigwidth}
		\centering\hfill
		\raisebox{0.25\twofigwidth}{
			\includegraphics[width=0.3\twofigwidth]{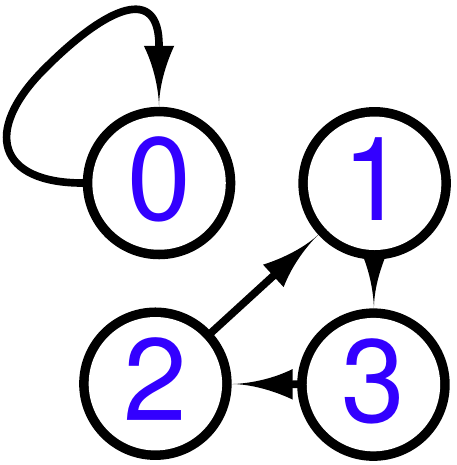}}
		a)
	\end{minipage}\hspace{6em}
	\begin{minipage}[t]{0.65\twofigwidth}
		\centering
		\raisebox{0.15\twofigwidth}{\hfill
			\includegraphics[width=0.65\twofigwidth]{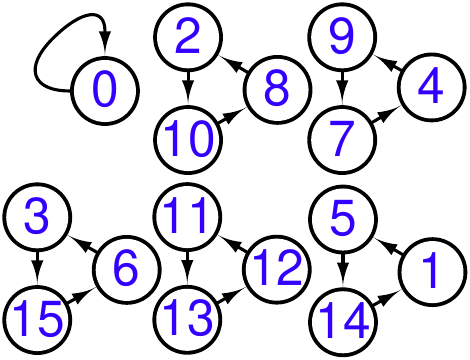}}
		b)
	\end{minipage} \hspace{6em}
	\begin{minipage}[t]{\twofigwidth}
		\centering
		\includegraphics[width=\twofigwidth]{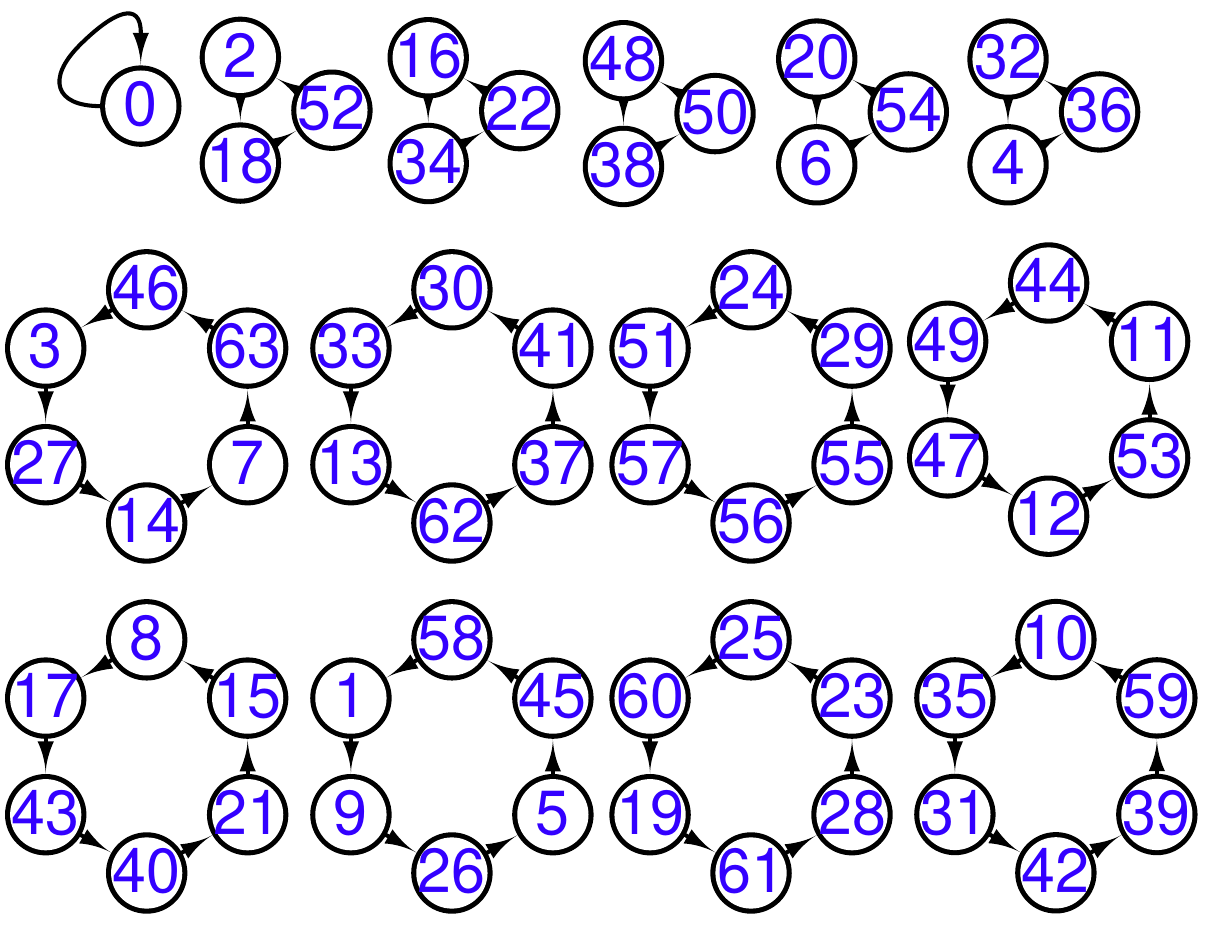}
		c)
	\end{minipage}\vspace{0.2em}
	\begin{minipage}{1.8\BigOneImW}
		\centering
		\includegraphics[width=1.8\BigOneImW]{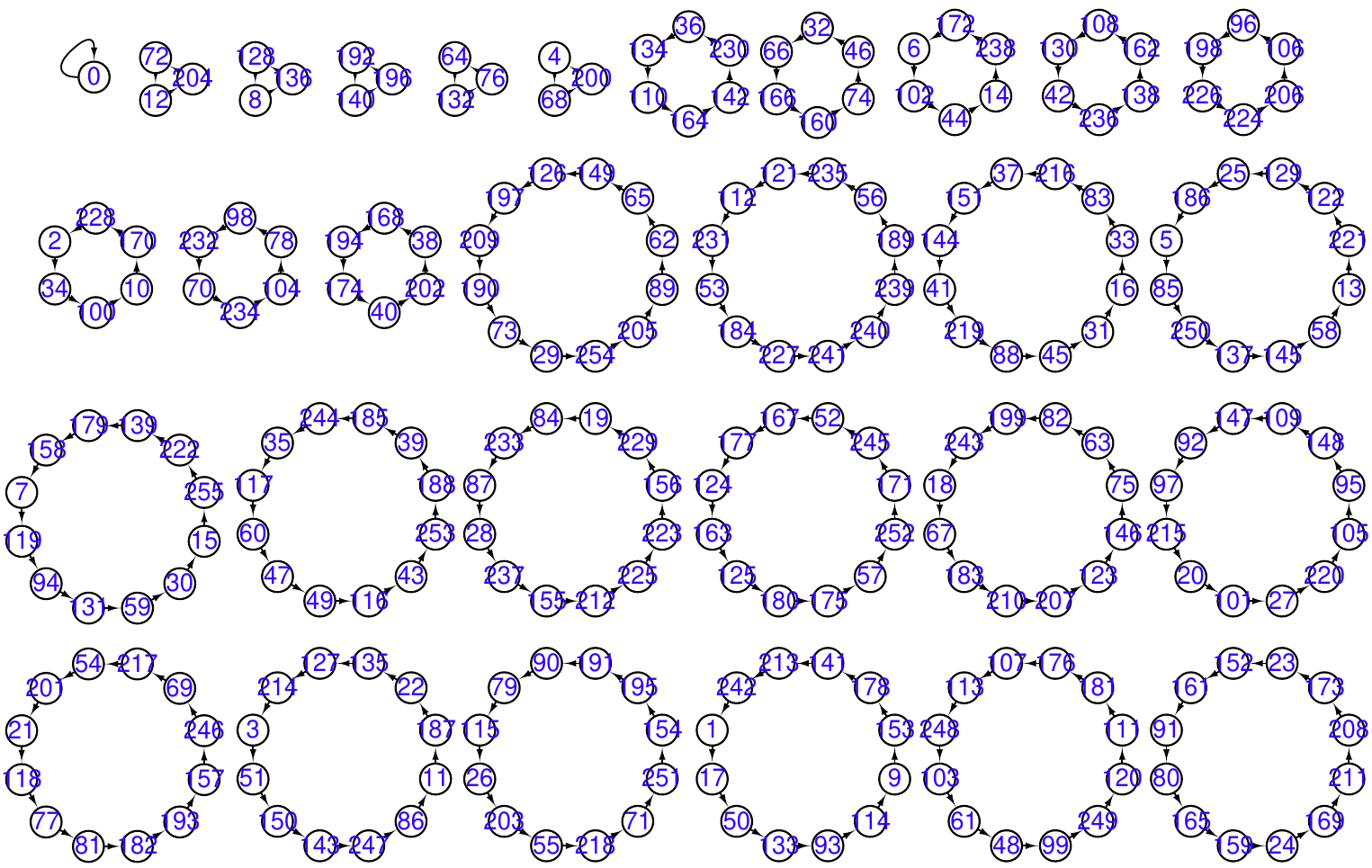}
		d)
	\end{minipage}
	\caption{Functional graphs of generalized Arnold's Cat maps in $\mathbb{Z}_{2^e}$ where $(p, q)=(1, 1)$:
		a) $e=1$; b) $e=2$; c) $e=3$; d) $e=4$.}
\label{fig:SMNcat}
\end{figure*}

Functional graph of Cat map (\ref{eq:ArnoldInteger}) can provide direct perspective on its structure.
The associate \emph{functional graph} $F_e$ can be built as follows: the $N^2$ possible states are viewed as $N^2$ nodes;
the node corresponding to $\textbf{x}_1=(x_1, y_1)$ is directly linked to the other one corresponding to $\textbf{x}_2=(x_2, y_2)$
if and only if $\textbf{x}_2=f(\textbf{x}_1)$ \cite{cqli:network:TCASI2019}. To facilitate visualization as a 1-D network data, every 2-D vector in Cat map~(\ref{eq:ArnoldInteger})
is transformed by a bijective function $z_n=x_{n} + (y_{n} \cdot N)$. To describe how the functional graph of Cat map~(\ref{eq:ArnoldInteger}) change with the arithmetic precision $e$, let
\begin{equation}
z_{n, e} = x_{n, e} + (y_{n, e} \cdot 2^e),
\label{eq:quantization}
\end{equation}
where $x_{n, e}$ and $y_{n, e}$ denote $x_{n}$ and $y_{n}$ of Cat map~(\ref{eq:ArnoldInteger}) with $N=2^e$, respectively.

As a typical example, we depicted the functional graphs of Cat map~(\ref{eq:ArnoldInteger}) with $(p, q)=(1, 1)$ in four domains $\{\mathbb{Z}_{2^e}\}_{e=1}^4$ in Fig.~\ref{fig:SMNcat}, where the number inside each circle (node) is $z_{n, e}$ in $F_e$. From Fig.~\ref{fig:SMNcat}, one can observe some general properties of functional graphs of Cat map~(\ref{eq:ArnoldInteger}). Especially, there are only cycles, no any transient. The properties on permutation are concluded in Properties~\ref{prop:bijective}, \ref{prop:onecycle}.

\begin{Property}
Cat map~(\ref{eq:ArnoldInteger}) defines a bijective mapping on the set $(0, 1, 2, \cdots, N^{2}-1)$.
\label{prop:bijective}
\end{Property}
\begin{proof}
As Cat map \eqref{eq:ArnoldInteger} is area-preserving on its domain, it defines a bijective mapping on $\mathbb{Z}^2_{N}$,
which is further transformed into a bijective mapping on $\mathbb{Z}_{N^2}$ by conversion function~(\ref{eq:quantization}).
\end{proof}

\begin{Property}
As for a given $N$, any node of functional graph of Cat map~(\ref{eq:ArnoldInteger}) belongs one and only one cycle, a set of nodes such that
Cat map~(\ref{eq:ArnoldInteger}) iteratively map them one to the other in turn.
\label{prop:onecycle}
\end{Property}
\begin{proof}
Referring to \cite[Theorem 5.1.1]{hall1959marshall}, the set $(0, 1, 2, \cdots, N^{2}-1)$ is divided into some disjoint subsets such that Cat map~(\ref{eq:ArnoldInteger}) is a cycle on each subset.
\end{proof}

\setlength\FourImW{0.13\columnwidth}
\begin{figure}[!htb]
	\centering
	\begin{minipage}{\FourImW}
		\centering
		\includegraphics[width=\FourImW]{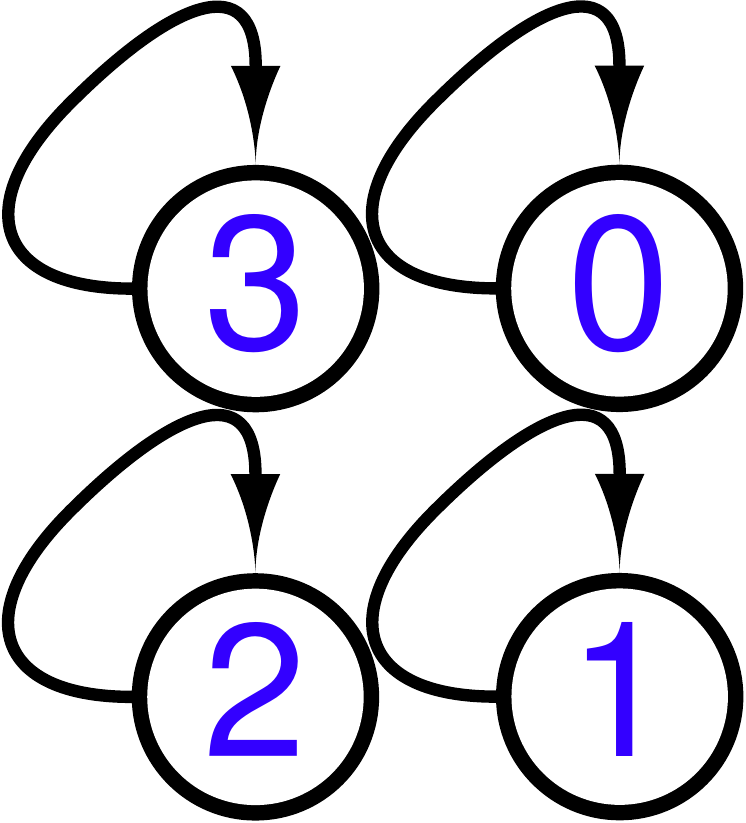}
		a)
	\end{minipage}\hspace*{4pt}
	\begin{minipage}{\FourImW}
		\centering
		\includegraphics[width=\FourImW]{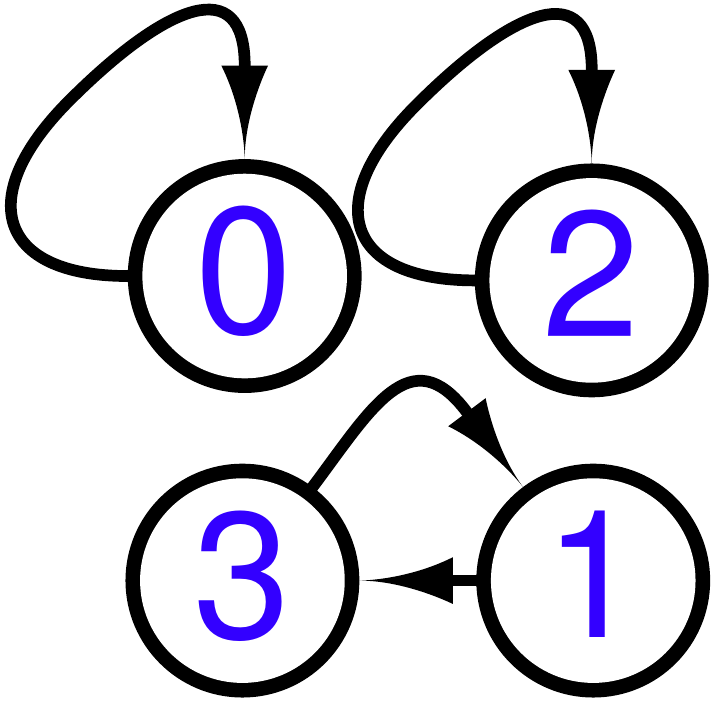}
		b)
	\end{minipage}\hspace*{4pt}
	\begin{minipage}{\FourImW}
		\centering
		\includegraphics[width=\FourImW]{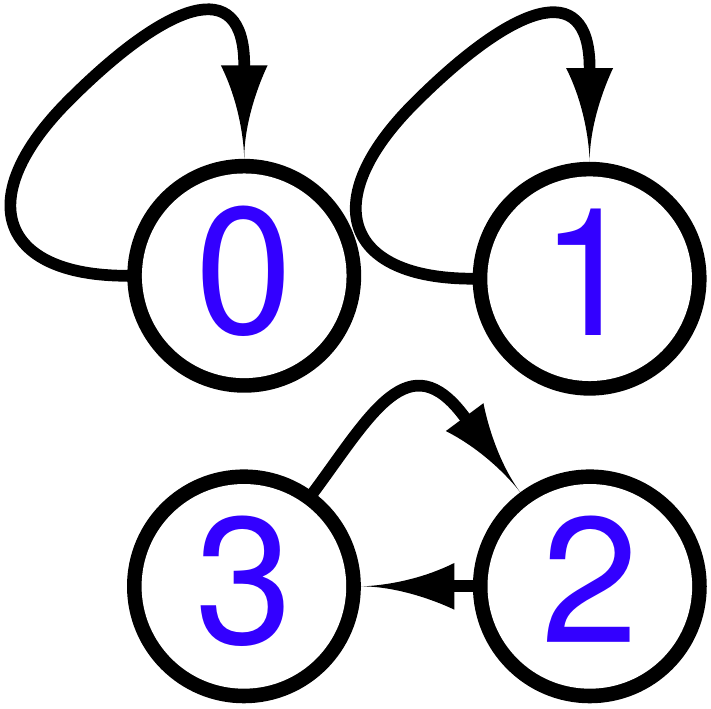}
		c)
	\end{minipage} \hspace*{4pt}
	\begin{minipage}{\FourImW}
		\centering
		\includegraphics[width=\FourImW]{a1_b1_e1}
		d)
	\end{minipage}
	\caption{Four possible functional graphs of Cat map~(\ref{eq:ArnoldInteger}) with $N=2$:
		a) $p$ and $q$ are both even; b) $p$ is even, and $q$ is odd;
		c) $p$ is odd, $q$ is even; d) $p$ and $q$ are both odd.}
	\label{fig:perioddistributione1}
\end{figure}

\setlength\FourImW{0.2\columnwidth}
\begin{figure}[!htb]
\centering
\begin{minipage}[b]{\FourImW}
	\centering
\includegraphics[keepaspectratio,width=\FourImW,height=\FourImW]{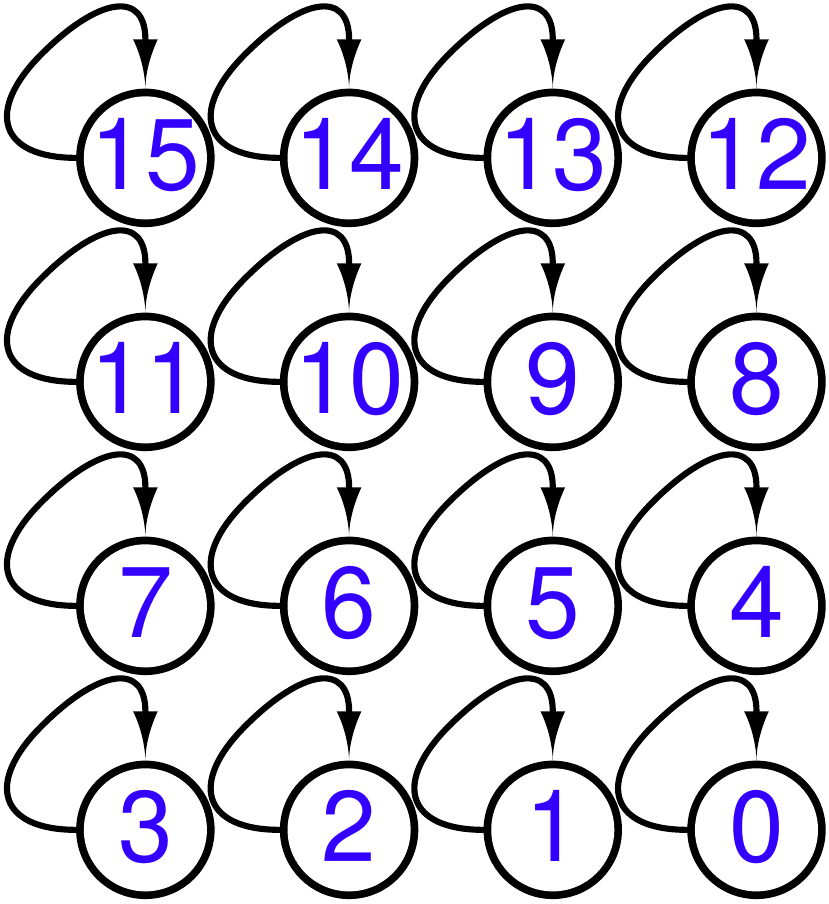}
\subcaption*{0)}
\end{minipage}
\begin{minipage}[b]{\FourImW}
	\centering
	\includegraphics[width=\FourImW]{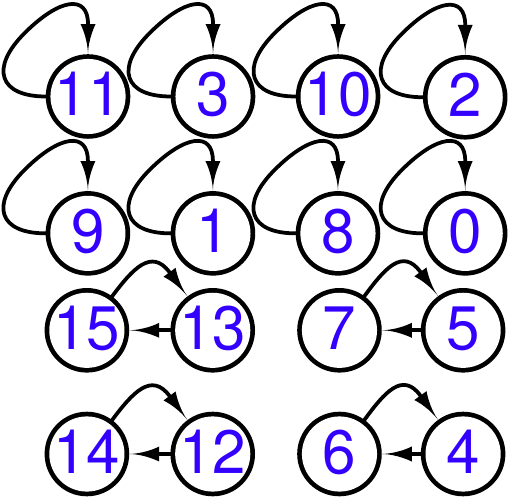}
\subcaption*{2)}
\end{minipage}
\begin{minipage}[b]{\FourImW}
	\centering
	\includegraphics[width=\FourImW]{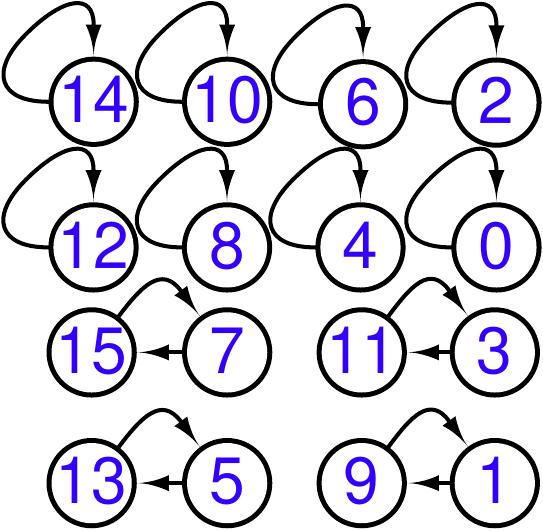}
\subcaption*{8)}
\end{minipage}
\begin{minipage}[b]{\FourImW}
	\centering
	\includegraphics[width=\FourImW]{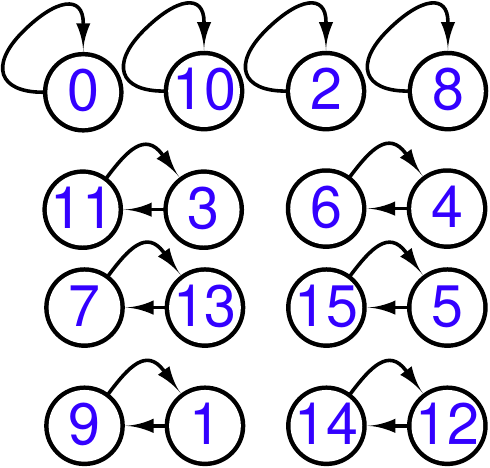}
\subcaption*{10)}
\end{minipage}\\
\begin{minipage}[b]{\FourImW}
	\centering
	\includegraphics[width=\FourImW]{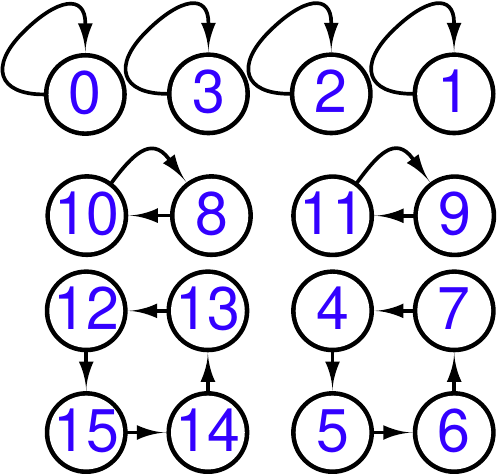}
\subcaption*{1)}
\end{minipage}
\begin{minipage}[b]{\FourImW}
	\centering
	\includegraphics[width=\FourImW]{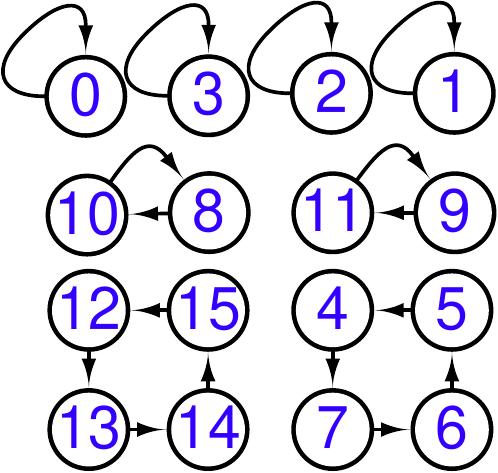}
\subcaption*{3)}
\end{minipage}
\begin{minipage}[b]{\FourImW}
	\centering
	\includegraphics[width=\FourImW]{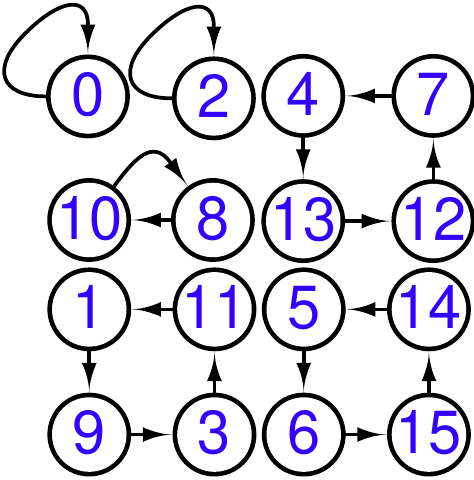}
\subcaption*{9)}
\end{minipage}
\begin{minipage}[b]{\FourImW}
	\centering
	\includegraphics[width=\FourImW]{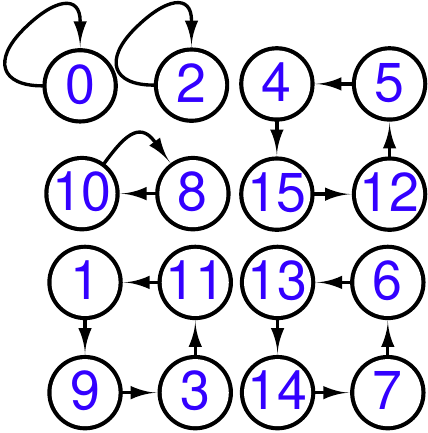}
\subcaption*{11)}
\end{minipage}  \\
\begin{minipage}[b]{\FourImW}
	\centering
	\includegraphics[width=\FourImW]{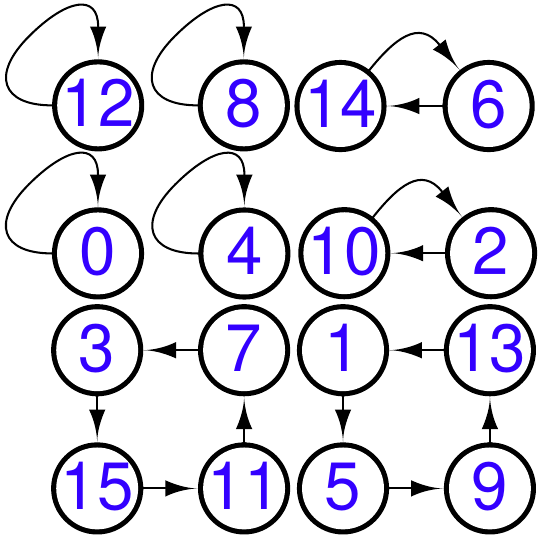}
\subcaption*{4)}
\end{minipage}
\begin{minipage}[b]{\FourImW}
	\centering
	\includegraphics[width=\FourImW]{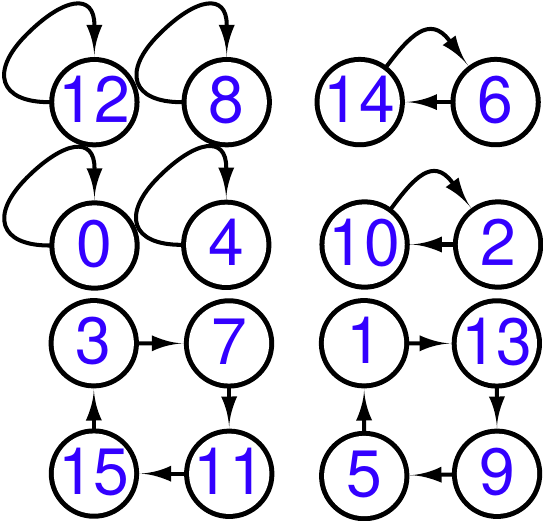}
\subcaption*{12)}
\end{minipage}
\begin{minipage}[b]{\FourImW}
	\centering
	\includegraphics[width=\FourImW]{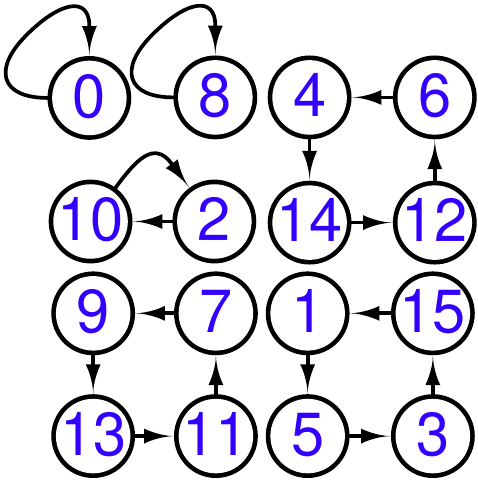}
\subcaption*{6)}
\end{minipage}
\begin{minipage}[b]{\FourImW}
	\centering
	\includegraphics[width=\FourImW]{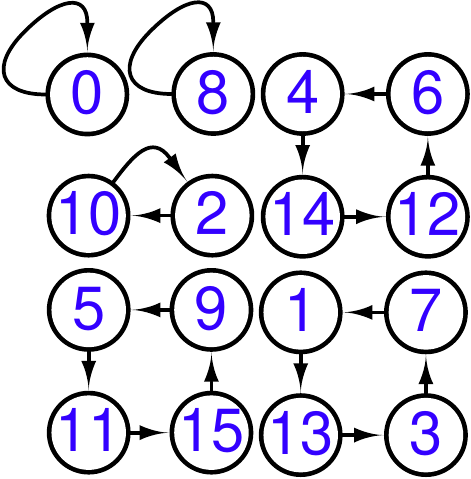}
\subcaption*{14)}
\end{minipage}  \\
\begin{minipage}[b]{\FourImW}
\centering
\includegraphics[width=\FourImW]{a1_b1_e2}
\subcaption*{5)}
\end{minipage}
\begin{minipage}[b]{\FourImW}
	\centering
	\includegraphics[width=\FourImW]{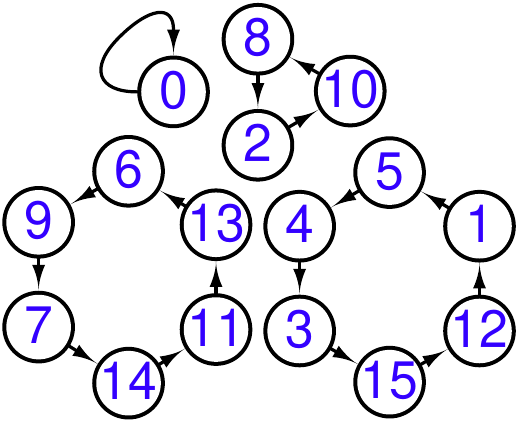}
\subcaption*{7)}
\end{minipage}
\begin{minipage}[b]{\FourImW}
\centering
\includegraphics[width=\FourImW]{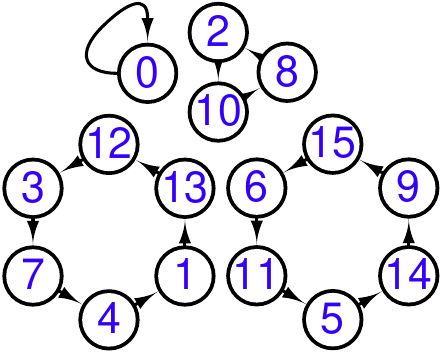}
\subcaption*{13)}
\end{minipage}
\begin{minipage}[b]{\FourImW}
\centering
\includegraphics[width=\FourImW]{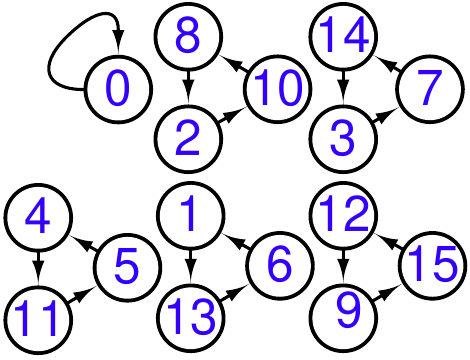}
\subcaption*{15)}
\end{minipage}
\caption{All possible functional graphs of Cat map~(\ref{eq:ArnoldInteger}) with $N=2^2$, where
the subfigure with caption ``$i)$" is corresponding to $(p, q)$ satisfying $i=p\bmod 4+(q\bmod 4)\cdot 4$.}
\label{Functionalgraphse1}
\end{figure}

As the period of a Cat map in a domain is the least common multiple of the periods of its cycles, the functional graph of a Cat map possessing a large period may be composed of a great number of cycles of very small periods. The whole graph shown in Fig.~\ref{fig:SMNcat}d) is composed of 16 cycles of period 12, 10 cycles of period 6, 1 cycle of period 3, and 1 self-connected cycle.

We found that there exists strong evolution relationship between $F_e$ and $F_{e+1}$. A node $z_{n, e}=x_{n, e}+y_{n, e}2^e$ in $F_e$ is evoluted to
\begin{IEEEeqnarray}{rCl}
	z_{n, e+1} & = & (x_{n, e}+a_{n}2^e) + (y_{n, e}+b_{n}2^e) 2^{e+1} \nonumber\\
	& = & z_{n, e}+(a_{n}2^e+y_{n, e}2^e+b_{n}2^{2e+1} ),
	\label{eq:evoluteNumber}	
\end{IEEEeqnarray}
where $a_{n}, b_{n}\in \{0, 1\}$. The relationship between iterated node of $z_{n, e}$ in $F_e$ and the corresponding evoluted one in $F_{e+1}$ is described in Property~\ref{Prop:evolution}. Furthermore, the associated cycle is expanded to up to four cycles as presented in Property~\ref{Prop:cycleExpan}. Assign $(a_{n_0}, b_{n_0})$ with one element in set (\ref{setFour}), one can obtain the corresponding cycle in $F_{e+1}$ with the steps given in Property~\ref{Prop:cycleExpan}. Then, the other element in set (\ref{setFour}) can be assigned to $(a_{n_0}, b_{n_0})$ if it does not ever exist in the set in Eq.~(\ref{numberElements}) corresponding to every assigned value of $(a_{n_0}, b_{n_0})$.
Every cycle corresponding to different $(a_{n_0}, b_{n_0})$ can be generated in the same way.

\begin{Property}
If the differences between inputs of Cat map~(\ref{eq:ArnoldInteger}) with $N=2^e$ and that of
Cat map~(\ref{eq:ArnoldInteger}) with $N=2^{e+1}$ satisfy
	\begin{equation}
	\begin{bmatrix}
	x_{n, e+1}-x_{n, e} \\
	y_{n, e+1}-y_{n, e}
	\end{bmatrix}
	=
	\begin{bmatrix}
	a_{n}\\
	b_{n}
	\end{bmatrix} \cdot 2^e,
	\label{eq:condition}
	\end{equation}
	one has
	\begin{equation}
	\begin{bmatrix}
	a_{n+1}\\
	b_{n+1}
	\end{bmatrix}
	=
	\left[\begin{bmatrix}
	1 & p    \\
	q & 1+p\cdot q
	\end{bmatrix}\cdot
	\begin{bmatrix}
	a_{n}  \\
	b_{n}
	\end{bmatrix}+
	\begin{bmatrix}
	k_x \\
	k_y
	\end{bmatrix}\right] \bmod 2,
	\label{eq:kxky}
	\end{equation}
	$a_{n}, b_{n}\in \{0, 1\}$,
	$k_x=\lfloor k'_x/2^e\rfloor$, $k_y=\lfloor k'_y/2^e\rfloor$, and	
	\begin{equation}
	\begin{bmatrix}
	k'_x \\
	k'_y
	\end{bmatrix}
	=
	\begin{bmatrix}
	1 & p    \\
	q & 1+p\cdot q
	\end{bmatrix}\cdot
	\begin{bmatrix}
	x_{n, e} \\
	y_{n, e}
	\end{bmatrix}.
	\label{eq:k2xk2y}
	\end{equation}
	\label{Prop:evolution}
\end{Property}
\setlength{\arraycolsep}{2pt}   
\begin{proof}
	According to the linearity of Cat map~(\ref{eq:ArnoldInteger}), one can get
	\begin{IEEEeqnarray}{rCl}
		\IEEEeqnarraymulticol{3}{l}{
			\begin{bmatrix}
				x_{n+1, e+1}-x_{n+1, e} \\
				y_{n+1, e+1}-x_{n+1, e}
			\end{bmatrix}
		}\nonumber 	 \\
		& = &
		\left[
		\begin{bmatrix}
			1 & p    \\
			q & 1+p\cdot q
		\end{bmatrix}
		\begin{bmatrix}
			x_{n, e+1}-x_{n, e} \\
			y_{n, e+1}-y_{n, e}
		\end{bmatrix}
		+2^e
		\begin{bmatrix}
			k_x \\
			k_y
		\end{bmatrix}\right]
		\bmod{2^{e+1}}.
		\IEEEeqnarraynumspace\label{eq:ArnoldIntegerdiff}
	\end{IEEEeqnarray}
	As $k\cdot a\equiv k\cdot a'\Mod{m}$ if and only if $a\equiv a'\Mod{\frac{m}{\gcd(m, k)}}$ and $a_{n+1, e}, b_{n+1, e}\in \{0, 1\}$,
	the property can be proved by putting condition~(\ref{eq:condition}) into   equation (\ref{eq:ArnoldIntegerdiff}) and dividing its both sides and the modulo by $2^e$.
\end{proof}

\begin{Property}
	Given a cycle $\textbf{Z}_e=\{z_{n, e}\}_{n=0}^{T_c-1}=\{(x_{n, e}$, $y_{n, e})\}_{n=0}^{T_c-1}$ in $F_e$ and its any point $z_{n_0, e}$, one has that the cycle to which $z_{n_0, e+1}$ belongs in $F_{e+1}$ is
	\begin{equation*}
	\textbf{Z}_{e+1}=\left\{ z_{n, e+1}) \right\}_{n=n_0}^{n_0+kT_c-1},
	\end{equation*}
	where	
	\begin{multline}
	k =  \#\{(a_{n_0}, b_{n_0}), (a_{n_0+T_c}, b_{n_0+T_c}), (a_{n_0+2T_c}, b_{n_0+2T_c}),\\
	(a_{n_0+3T_c}, b_{n_0+3T_c})\},
	\label{numberElements}	
	\end{multline}
	$z_{n, e}=z_{n', e}$ for $n\ge T_c$, $n'=n\bmod T_c$, $\{(a_{n}, b_{n})\}_{n=n_0}^{n_0+3T_c}$ are generated by iterating Eq.~\eqref{eq:kxky} for $n=n_0\sim n_0+3T_c$, and $\#(\cdot)$ returns the ardinality of a set.
	\label{Prop:cycleExpan}
\end{Property}
\begin{proof}
	Given a node in a cycle, $(k_x, k_y)$ in Eq.~(\ref{eq:kxky}) is fixed, so
	Eq.~(\ref{eq:kxky}) defines a bijective mapping on set
	\begin{equation}
	\{(0, 0), (0, 1), (1, 0), (1, 1) \}
	\label{setFour}
	\end{equation}
	as shown in	Fig.~\ref{fig:rootmap}. So, $z_{n, e+1}$ may fall in set $\{z_{j, e+1}\}_{j=n_0}^n$
	when and only when $(n-n_0)\bmod T_c=0$ and $n>n_0$, i.e. the given cycle is went through one more times.
\end{proof}

\setlength\FourImW{0.24\columnwidth}
\begin{figure}[!htb]
\centering
\begin{minipage}{\FourImW}
\centering
\includegraphics[width=\FourImW]{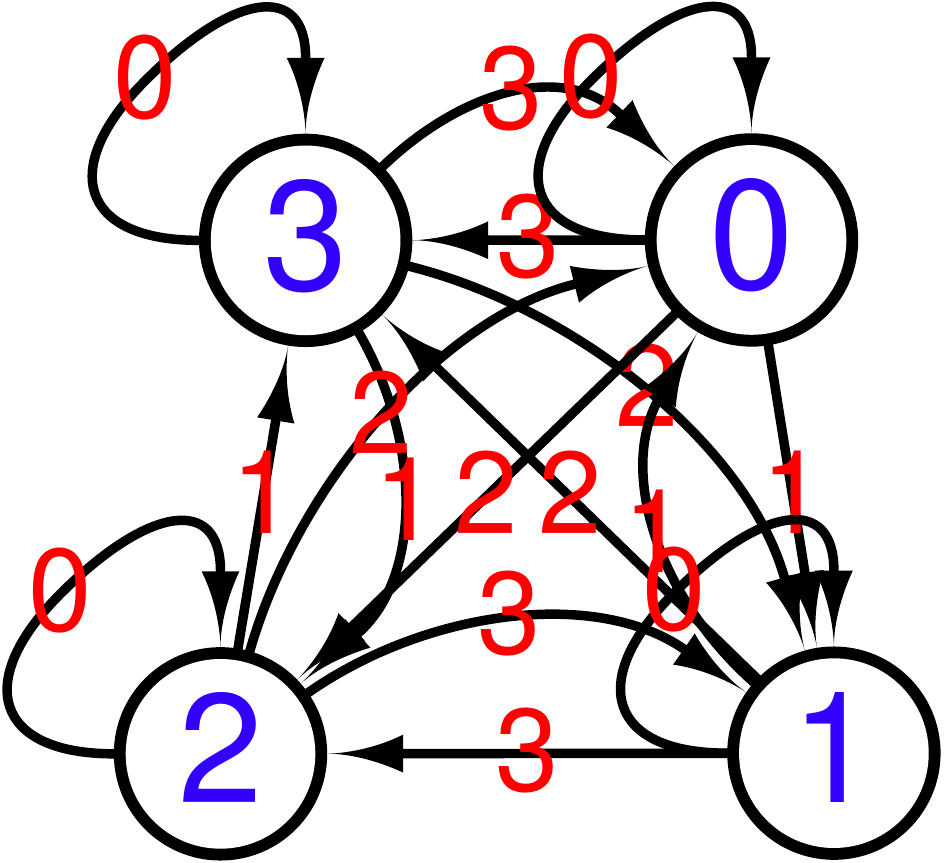}
a)
\end{minipage}
\begin{minipage}{\FourImW}
\centering
\includegraphics[width=\FourImW]{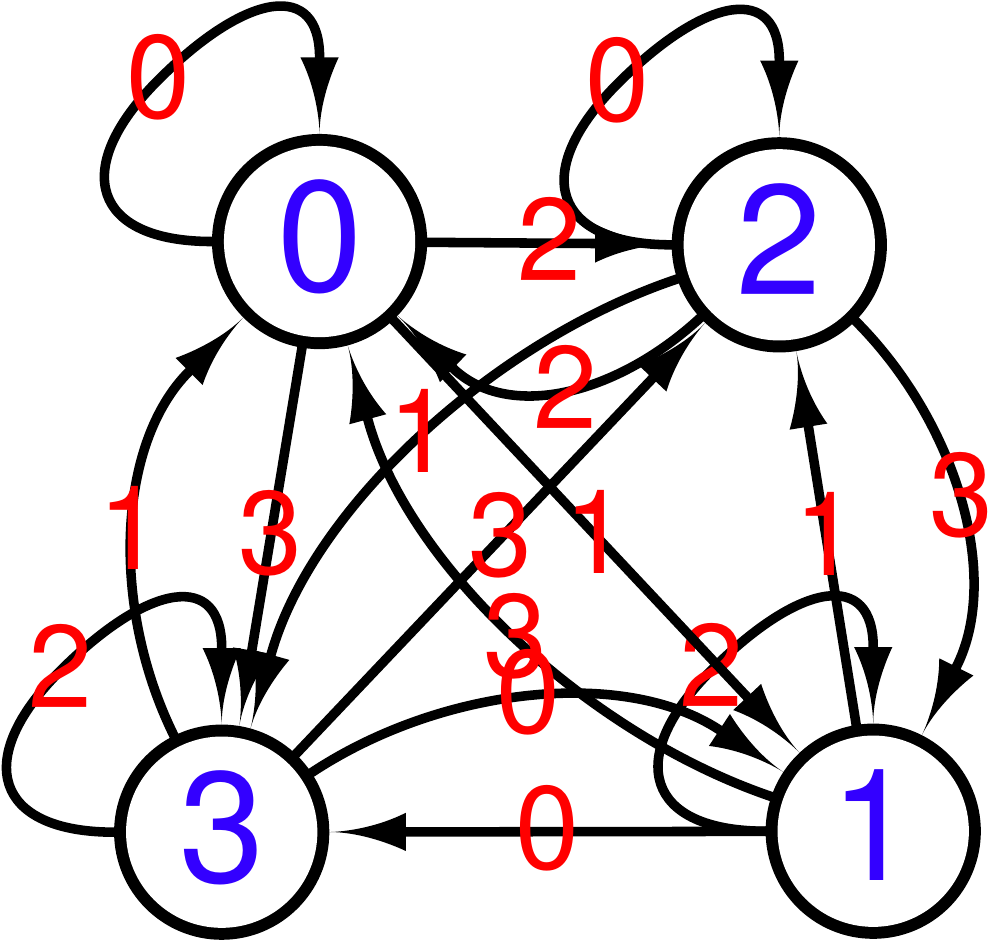}
b)
\end{minipage}
\begin{minipage}{\FourImW}
\centering
\includegraphics[width=\FourImW]{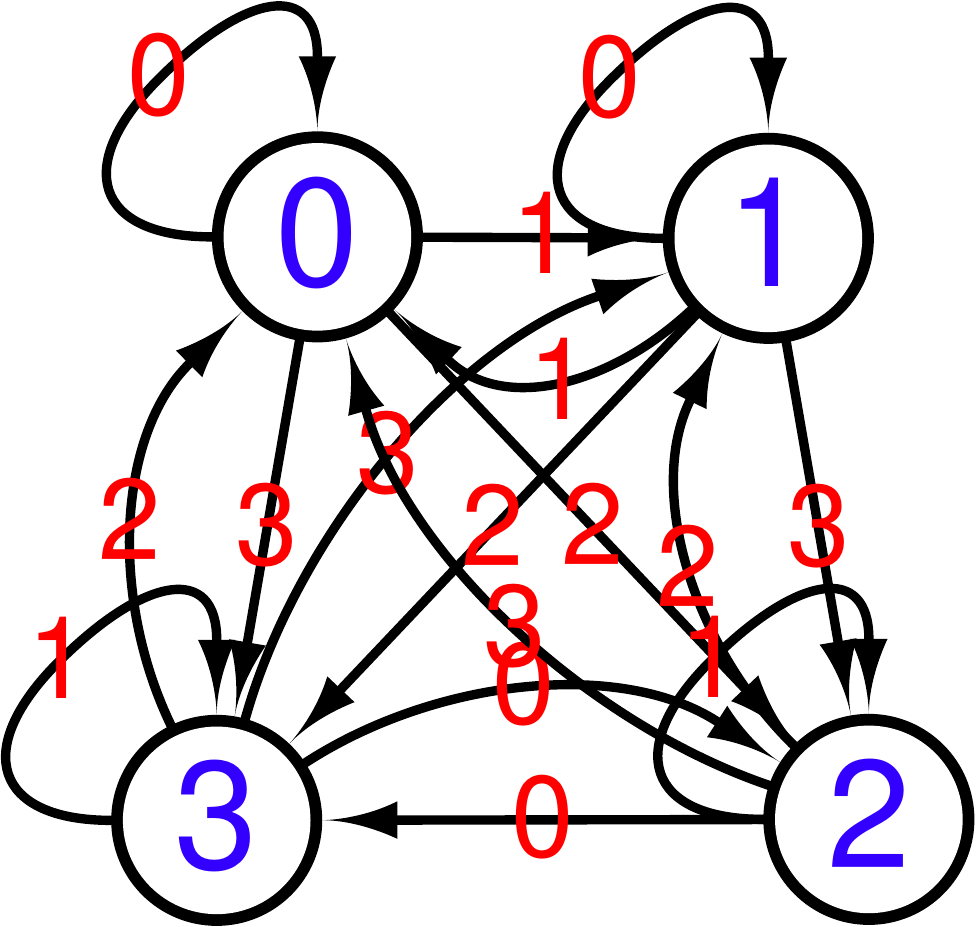}
c)
\end{minipage}
\begin{minipage}{\FourImW}
\centering
\includegraphics[width=\FourImW]{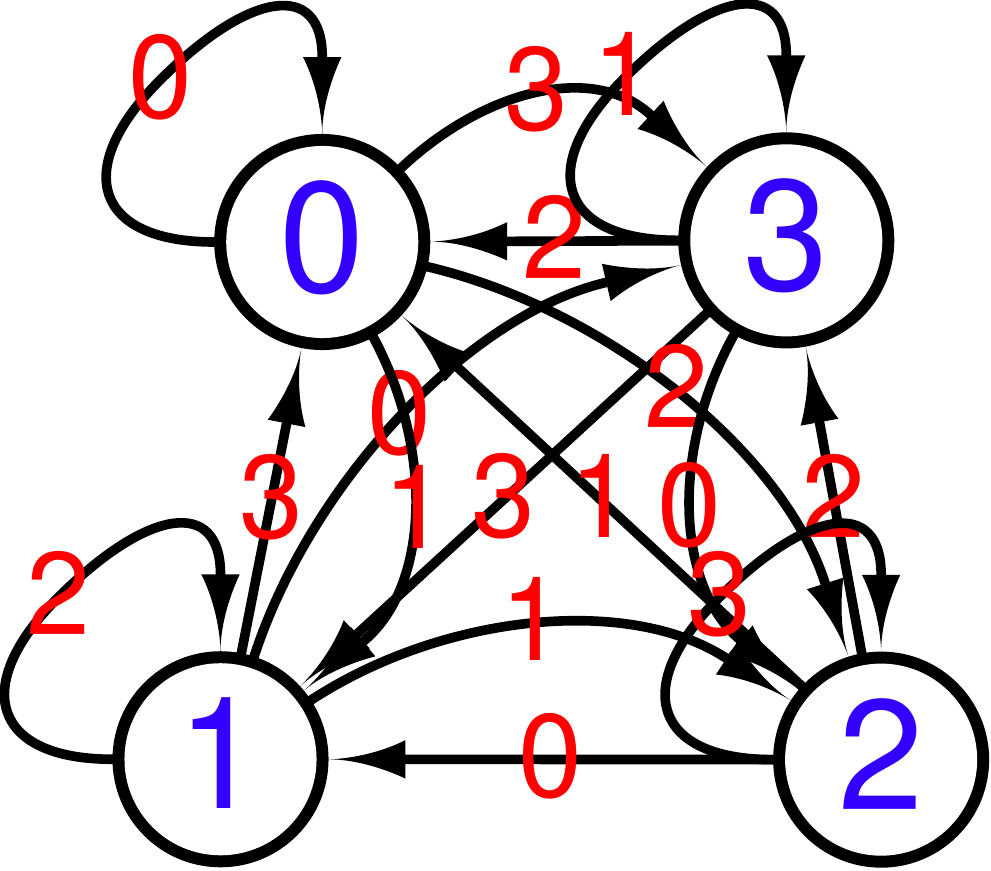}
d)
\end{minipage}
\caption{Mapping relationship between $(a_n+2b_n)$ and $(a_{n+1}+2b_{n+1})$ in Eq.~(\ref{eq:kxky}) with $(k_x+2k_y)$ shown beside the arrow:
	a) $p$ and $q$ are both even; b) $p$ is even, and $q$ is odd;
	c) $p$ is odd, $q$ is even; d) $p$ and $q$ are both odd.}
\label{fig:rootmap}
\end{figure}

Depending on the number of candidates for $(a_{n_0}, b_{n_0})$ and the corresponding cardinality in Eq.~(\ref{numberElements}), a cycle of length $T_c$ in $F_e$ is expanded to five possible cases in $F_{e+1}$:
1) one cycles of length $T_c$ and one cycle of length $3T_c$,
e.g. the self-connected cycle in Fig.~\ref{fig:SMNcat}a), $``0\rightarrow 0"$, is evoluted to two cycles in Fig.~\ref{fig:SMNcat}b),
$``0\rightarrow 0"$ and ``$(0+2^1)=2\rightarrow (0+2^3)=8\rightarrow (0+2^1+2^3)=10\rightarrow 2$";
2) twos cycles of length $T_c$ and one cycle of length $2T_c$, e.g. the cycle $``0\rightarrow 0"$ in Fig.~\ref{fig:perioddistributione1}c) is expanded to three cycles in the same SMN: $``0\rightarrow 0"$; $``2\rightarrow 2"$; $``8\rightarrow 10\rightarrow 8"$;
3) four cycles of length $T_c$, e.g.
the cycle $``1\rightarrow 13 \rightarrow 12 \rightarrow 3 \rightarrow 7 \rightarrow 4 \rightarrow 1"$ in Fig.~\ref{fig:SMNcat}b) is expanded to four cycles of the same length shown in the lower left side of Fig.~\ref{fig:SMNcat}c).
4) two cycles of length $2T_c$, e.g.
the cycle in Fig.~\ref{fig:SMNcat}a), $``1\rightarrow 3\rightarrow 2\rightarrow 1"$ is evoluted to the other two cycles in Fig.~\ref{fig:SMNcat}b),
``$1 \rightarrow (3+2^1+2^3)=13 \rightarrow (2+2^3+2)=12 \rightarrow (1+2^1)=3 \rightarrow (3+2^1+2)=7
\rightarrow (2+2^1)=4 \rightarrow 1$"; ``$(1+2^3)=9\rightarrow (3+2^1+2^3+2)=15 \rightarrow (2+2^1+2)=6 \rightarrow (1+2^1+2^3)=11 \rightarrow (3+2^1)=5
\rightarrow (2+2^1+2^3+2)=14 \rightarrow 9$";
5) one cycle of length $4T_c$, e.g.
the cycle $``1\rightarrow 1"$ in Fig.~\ref{fig:perioddistributione1}c) is expanded to $``1\rightarrow 9 \rightarrow 3 \rightarrow 11 \rightarrow 1"$ in the subfigure with caption ``9)" in Fig.~\ref{Functionalgraphse1}.
In all, all the fives possible cases can be found in Fig.~\ref{fig:SMNcat}, \ref{fig:perioddistributione1}, \ref{Functionalgraphse1}.

As shown in Property~\ref{Prop:cycleExpan}, any cycle of $F_{e+1}$ is incrementally expanded from a cycle of $F_1$. So, the number of cycles of a given length in $F_{e+1}$ has some relationship with that of the corresponding length in $F_{e}$, which is determined by the control parameters $p$, $q$. Moreover, as shown in Property~\ref{prop:isomorphism}, $F_e$ is isomorphic to a part of $F_{e+1}$, which can be verified in Fig.~\ref{fig:SMNcat}.

\begin{Property}
Any cycle $\{(\textbf{C}^i\cdot\textbf{X})\bmod 2^{e}\}_{i=1}^{T_c}$ in $F_e$ and
	the corresponding cycle $\{(\textbf{C}^i\cdot(2\textbf{X}))\bmod 2^{e+1}\}_{i=1}^{T_c}$ in $F_{e+1}$ compose two isomorphic groups with respect to their respective operators.
	\label{prop:isomorphism}
\end{Property}
\begin{proof}
	As for any point $\textbf{X}$ in $F_e$, define a multiplication operation $\circ$ for any two elements of set $G=\{(\textbf{C}^i\cdot\textbf{X})\bmod 2^{e}\}_{i=1}^{T_c}$,
	$g_1 \circ g_2=(\textbf{C}^{i_1+i_2}\cdot\textbf{X})\bmod 2^e$, where $g_1=(\textbf{C}^{i_1}\cdot\textbf{X})\bmod 2^e$,
	$g_2=(\textbf{C}^{i_2}\cdot\textbf{X})\bmod 2^e$.
	The set $G$ is closed with respect to the operator $\circ$.
	Point $(\textbf{C}^{T_c}\cdot\textbf{X})\bmod 2^e=(\textbf{C}^{0}\cdot\textbf{X})\bmod 2^e=\textbf{X}$ is the identity element.
	Multiplication of any three matrices satisfy the associative law. As for any element $g_1$, there is an inverse element
	$(\textbf{C}^{T_c-i_1}\cdot\textbf{X})\bmod 2^e$. So, the non-empty
	set $G$ composes a group with respect to the operator.
	Referring to the elementary properties of congruences summarized in \cite[P.61]{hardy1979numbers}, equation
	\begin{equation*}	
	\textbf{C}^{i}\cdot
	\begin{bmatrix}
	x_{0} \\
	y_{0}
	\end{bmatrix}\bmod 2^e
	=
	\begin{bmatrix}
	x_{0} \\
	y_{0}
	\end{bmatrix}
	\end{equation*}
	holds if and only if
	\begin{equation*}	
	\textbf{C}^{i}\cdot
	\begin{bmatrix}
	2x_{0} \\
	2y_{0}
	\end{bmatrix}\bmod 2^{e+1}
	=
	\begin{bmatrix}
	2x_{0} \\
	2y_{0}
	\end{bmatrix}.
	\end{equation*}
	So $G'=\{\textbf{C}^i\cdot(2\textbf{X})\bmod 2^{e+1}\}_{i=1}^{T_c}$ also composes
	a group with respect to operator $\hat{\circ}$,
	where $g'_1 \hat{\circ} g'_2=(\textbf{C}^{i_1+i_2}\cdot(2\textbf{X}))\bmod 2^{e+1}$, $g'_1=(\textbf{C}^{i_1}\cdot(2\textbf{X}))\bmod 2^{e+1}$,
	$g'_2=(\textbf{C}^{i_2}\cdot(2\textbf{X}))\bmod 2^{e+1}$.
	Therefor, the two groups are isomorphic with respect to bijective map $y=(2\textbf{X})\bmod 2^{e+1}$.
\end{proof}

The period distribution of cycles in $F_e$ follows a power-law distribution of fixed exponent one when $e$ is sufficiently large. The number of cycles of any length is monotonously increased to a constant with respect to $e$, which is shown in Table~\ref{table:pqeTc}, where the dashline marked the case corresponding to the threshold value.

In \cite{Catchen2013period2e}, it is assumed that $e\ge 3$ ``because the cases when $e=1$ and $e=2$ are trivial". On the contrary, the structure of functional graph of Cat map~(\ref{eq:ArnoldInteger}) with $e=1$, shown in Fig.~\ref{fig:perioddistributione1}, plays a fundamental role for that with $e\ge 3$, e.g. the cycles of length triple of 3 in $F_e$ (if there exist) are generated by the cycle of length 3 in $F_1$.

\setlength{\arraycolsep}{4pt}   
\setlength\tabcolsep{4pt} 
\addtolength{\abovecaptionskip}{-2pt}
\begin{table}[!htb]
\caption{The number of cycles of period $T_c$ in $F_e$ with $(p, q)=(9, 14)$.} 
\centering 
\begin{tabular}{*{8}{c|}c} 
\hline 
\diagbox[width=6em]{$e$}{$N_{T_c, e}$}{$T_c$}
&$2^0$ & $2^1$ &     $2^2$&     $2^3$&    $2^4$ &    $2^5$ &     $2^6$&    $2^7$ \\ \hline
1 &2 &1&     0&     0&     0&     0&      0&      0  \\
2 &2 &1&     3&     0&     0&     0&      0&      0  \\
3 &2 &1&    15&     0&     0&     0&      0&      0  \\
4 &2 &1&    63&     0&     0&     0&      0&      0  \\
5 &2 &1&   255&     0&     0&     0&      0&      0  \\
6 &2 &1&  1023&     0&     0&     0&      0&      0  \\
7 &2 &1&  4095&     0&     0&     0&      0&      0  \\
8 &2 &1& 16383&     0&     0&     0&      0&      0  \\ \hdashline[2pt/1pt] 
9 &2 &1& 16383& 24576&     0&     0&      0&      0  \\ \hdashline[5pt/2pt]
10&2 &1& 16383& 24576& 49152&     0&      0&      0  \\
11&2 &1& 16383& 24576& 49152& 98304&      0&      0  \\
12&2 &1& 16383& 24576& 49152& 98304& 196608&      0  \\
13&2 &1& 16383& 24576& 49152& 98304& 196608& 393216  \\\hline 
\end{tabular}
\label{table:pqeTc} 
\end{table}

\subsection{Properties on iterating Cat map over $(\mathbb{Z}_{N}, +,\ \cdot\ )$}
\label{ssec:iterateCat}

Diagonalizing the transform matrix of Cat map~(\ref{eq:ArnoldInteger}) with its eigenmatrix,  the explicit representation of $n$-th iteration of the map can be obtained as Theorem~\ref{theorem:Cat}, which serves as basis of the analysis of this paper.

The necessary and sufficient condition for the least period of Cat map~\eqref{eq:ArnoldInteger} over $(\mathbb{Z}_{N}, +,\ \cdot\ )$ is given in
Proposition~\ref{prop:leastPeriod}. Considering the even parity of $G_n$, the condition
can be simplified as Corollary~\ref{co:pqhGn2}. Based on the property of $H_n$ in Lemma~\ref{le:HG}, the inverse of the $n$-th iteration of Cat map is obtained as shown in
Proposition~\ref{prop:inverse}.

\begin{theorem}
The $n$-th iteration of Cat map matrix \eqref{eq:MatMatrix} satisfies
	\begin{equation}
	\textbf{C}^n=
	\begin{bmatrix}
	\frac{1}{2}G_n-\frac{A-2}{2}H_n & p\cdot H_n    \\
	q\cdot H_n                      & \frac{1}{2}G_n+\frac{A-2}{2}H_n
	\end{bmatrix},
	\label{eq:iterateCat}
	\end{equation}
	where
	\begin{equation}
	\begin{cases}
	G_n=(\frac{A+B}{2})^n+(\frac{A-B}{2})^n, \\
	H_n=\frac{1}{B}( (\frac{A+B}{2})^n-(\frac{A-B}{2})^n ),
	\end{cases}
	\label{eq:GH}
	\end{equation}
	$B=\sqrt{A^2-4}$ and $A=p\cdot q+2$.
\label{theorem:Cat}
\end{theorem}
\begin{proof}
	First, one can calculate the characteristic polynomial of
	Cat map matrix~\eqref{eq:MatMatrix} as
	\begin{IEEEeqnarray*}{rCl}
		|\textbf{C}-\lambda \textbf{I}| & = &
		\det\begin{bmatrix}
			1-\lambda & p \\
			q         & p\cdot q+1-\lambda
		\end{bmatrix}    \\
		& = & \lambda^2-(p\cdot q+2)\lambda+1 \\	
		& = & 0.
	\end{IEEEeqnarray*}
	Solving the above equation, one can obtain two
	characteristic roots of Cat map matrix:
	\begin{equation*}
	\left\{
	\begin{split}
	\lambda_1 & =\frac{A+B}{2}, \\
	\lambda_2 & =\frac{A-B}{2}.
	\end{split}
	\right.
	\end{equation*}
	Setting $\lambda$ in $(\textbf{C}-\lambda \textbf{I})\cdot\textbf{X}=0$ as $\lambda_1$ and $\lambda_2$ separately, the corresponding eigenvector
	$\xi_{\lambda_1}=[1, \frac{ A-2+B}{2p}]^\intercal$ and $\xi_{\lambda_2}=[1, \frac{ A-2-B}{2p}]^\intercal$ can be obtained, which means
	\begin{equation}
	\label{eq:commu}
	\textbf{C}\cdot \textbf{P}=\textbf{P}\cdot \Lambda,
	\end{equation}
	where
	$\textbf{P}=(\xi_{\lambda_1}, \xi_{\lambda_2})$ and
	\begin{equation*}	
	\Lambda=\begin{bmatrix}
	\lambda_1 & 0 \\
	0         & \lambda_2
	\end{bmatrix}.
	\end{equation*}
	From Eq.~(\ref{eq:commu}), one has
	\begin{equation}
	\textbf{C} =\textbf{P}\cdot \Lambda\cdot \textbf{P}^{-1},
	\end{equation}
	where
	\begin{equation*}
	\textbf{P}^{-1}=\begin{bmatrix}
	-\frac{A-2-B}{2B} & \frac{p}{B} \\
	\frac{A-2+B}{2B}  & -\frac{p}{B}
	\end{bmatrix}.
	\end{equation*}
	Finally, one can get
	\begin{equation*}
	\begin{split}
	\textbf{C}^n &=(\textbf{P}\cdot \Lambda\cdot \textbf{P}^{-1})^n\\
	&=\textbf{P}\cdot \Lambda^n\cdot \textbf{P}^{-1}\\
	&=\begin{bmatrix}
	\frac{1}{2}G_n- \frac{A-2}{2}H_n & p H_n \\
	q H_n& \frac{1}{2}G_n+ \frac{A-2}{2}H_n
	\end{bmatrix},
	\end{split}
	\end{equation*}
	where $G_n$ and $H_n$ are defined as Eq.~\eqref{eq:GH}.
\end{proof}

\begin{Proposition}
The least period of Cat map~\eqref{eq:ArnoldInteger} over $(\mathbb{Z}_{N}, +,\ \cdot\ )$ is $T$ if and only if $T$ is the minimum possible value of $n$ satisfying	
\begin{equation}
\left\{ \,
\begin{IEEEeqnarraybox}[
	\IEEEeqnarraystrutmode
	\IEEEeqnarraystrutsizeadd{6pt}{6pt}][c]{rCl}
	p\cdot H_n         & \equiv & 0 \bmod N,    \\
    q\cdot H_n         & \equiv &0 \bmod N,     \\
\frac{1}{2}G_n-\frac{1}{2}p\cdot q\cdot H_n  & \equiv & 1\bmod N,  \\
\frac{1}{2}G_n+\frac{1}{2}p\cdot q\cdot H_n  & \equiv & 1\bmod N,
	\end{IEEEeqnarraybox}
	\right.
	\label{eq:pqhGn}
\end{equation}
\label{prop:leastPeriod}
\end{Proposition}	
\begin{proof}
If the period of Cat map~\eqref{eq:ArnoldInteger} over $(\mathbb{Z}_{N}, +,\ \cdot\ )$ is $T$,
$\textbf{C}^T \cdot\textbf{X}\equiv \textbf{X} \bmod N$ exists for any $\textbf{X}$. Setting $\textbf{X}=[1, 0]^\intercal$ and $\textbf{X}=[0, 1]^\intercal$ in order, one can get
\begin{equation}
\textbf{C}^n \equiv
\begin{bmatrix}
1   & 0   \\
0   & 1
\end{bmatrix} \bmod N
\label{eq:CnI}
\end{equation}
when $n=T$. Incorporating Eq.~(\ref{eq:iterateCat}) into the above equation, one can assure than Eq.~(\ref{eq:pqhGn}) exists when $n=T$. As $T$ is the least period, $T$ is the minimum possible value of $n$ satisfying Eq.~(\ref{eq:pqhGn}). The condition is therefore necessary.
If $T$ is the minimum possible value of $n$ satisfying Eq.~(\ref{eq:pqhGn}), Eq.~(\ref{eq:CnI}) holds, which means that
$T$ is a period of Cat map~\eqref{eq:ArnoldInteger} over $(\mathbb{Z}_{N}, +,\ \cdot\ )$. As
Eq.~(\ref{eq:CnI}) does no exist for any $n<T$, $T$ is the least period of Cat map~\eqref{eq:ArnoldInteger} over $(\mathbb{Z}_{N}, +,\ \cdot\ )$.
So the sufficient part of the proposition is proved.
\end{proof}

\begin{lemma}
	For any positive integer $m$, the parity of $G_{2^ms}$ is the same as that of $G_{s}$, and
	\[
	\frac{1}{2}G_{2^{m}\cdot s} \equiv 1 \bmod 2
	\]
	if $G_s$ is even, where $s$ is a given positive integer.
\label{lemma:parity}
\end{lemma}
\begin{proof}
	Referring to Eq.~(\ref{eq:GH}), one has
	\begin{IEEEeqnarray}{rCl}
		G_{2^m\cdot s} & =  &  \left(\frac{A+B}{2}\right)^{2^{m}\cdot s}+\left(\frac{A-B}{2}\right)^{2^{m}\cdot s}   \nonumber\\
		& =  &  \left(\left(\frac{A+B}{2}\right)^{2^{m-1}\cdot s}+\left(\frac{A-B}{2}\right)^{2^{m-1}\cdot s}\right)^2 \nonumber\\
		&    &  \;\;\;    -2 \left(\frac{A^2-B^2}{4}\right)^{2^{m-1}\cdot s}   \nonumber\\
		& =  &  (G_{2^{m-1}\cdot s})^2 -2.  \label{eq:Gsm+1}
	\end{IEEEeqnarray}
	When $m=1$, $G_{2s}=(G_{s})^2-2$. So the parity of $G_{2s}$ is the same as that of $G_{s}$ no matter $G_{s}$ is even or odd. In case $G_{s}$ is even,
	$\frac{1}{2}G_{2s}=\frac{1}{2}(G_{s})^2-1\equiv 1 \bmod 2$.	
	Proceed by induction on $m$ and assume that the lemma hold for any $m$ less than a positive integer $k>1$.
	When $m=k$, $G_{2^k\cdot s}=(G_{2^{k-1}\cdot s})^2-2$, which means that the parity of $G_{2^k\cdot s}$ is the same as that of $G_{2^{k-1}\cdot s}$ no matter $G_{2^{k-1}\cdot s}$ is even or odd. If $G_s$ is even,
	$\frac{1}{2}G_{2^k\cdot s}=\frac{1}{2}(G_{2^{k-1}\cdot s})^2-1 \equiv 1 \bmod 2$ also holds.
\end{proof}

\begin{Corollary}
If $G_n$ is even, condition~(\ref{eq:pqhGn}) is equivalent to
		\begin{equation}
		\left\{ \,
		\begin{IEEEeqnarraybox}[
		\IEEEeqnarraystrutmode
		\IEEEeqnarraystrutsizeadd{4pt}{4pt}][c]{rCl}
	p\cdot H_n         & \equiv & 0 \bmod N,    \\
    q\cdot H_n         & \equiv &0 \bmod N,     \\
	\frac{1}{2}p\cdot q\cdot H_n         & \equiv & 0 \bmod N, \\
	\frac{1}{2}	G_n & \equiv & 1 \bmod N. 	
		\end{IEEEeqnarraybox}
		\right.
		\label{eq:pqhGn2}	
		\end{equation}	
\label{co:pqhGn2}		
\end{Corollary}
\begin{proof}
If $G_n$ is even, $\frac{1}{2}G_n$ is an integer.
As $\frac{1}{2}G_n\pm \frac{1}{2}p\cdot q\cdot H_n$ is an integer,
$\frac{1}{2}G_n-\frac{1}{2}p\cdot q\cdot H_n$ is also an integer.
So, one can get $\frac{1}{2}p\cdot q\cdot H_n \equiv 0 \bmod N$
and $\frac{1}{2}	G_n \equiv 1 \bmod N$ from the last two congruences in condition~(\ref{eq:pqhGn}).
\end{proof}

\begin{lemma}	
Sequence $\{H_n\}_{n=1}^{\infty}$ satisfies
	\begin{equation}
	H_{2^m\cdot s}=H_{s}\cdot \prod_{j=0}^{m-1}G_{2^j\cdot s},
	\label{eq:Hs}
	\end{equation}
	where $m$ and $s$ are positive integers.
\label{le:HG}
\end{lemma}
\begin{proof}
	This Lemma is proved via mathematical induction on $m$.
	When $m=1$,
	\begin{IEEEeqnarray}{rCl}
		H_{2s} & =  &  \frac{1}{B}\left( \left(\frac{A+B}{2}\right)^{2s}-\left(\frac{A-B}{2}\right)^{2s} \right)\nonumber  \\
		& =  & H_s \cdot G_s.
	\label{eq:HsGs}
	\end{IEEEeqnarray}
	Now, assume the lemma is true for any $m$ in Eq.~(\ref{eq:Hs}) less than $k$.
	When $m=k$,
	\begin{IEEEeqnarray}{rCl}
		H_{2^k\cdot s} & = & \frac{1}{B}\left( \left(\frac{A+B}{2}\right)^{2^k\cdot s}-\left(\frac{A-B}{2}\right)^{2^k\cdot s} \right)   \nonumber\\
		& =  & \frac{1}{B}\left( \left(\frac{A+B}{2}\right)^{2^{k-1}\cdot s}-\left(\frac{A-B}{2}\right)^{2^{k-1}\cdot s} \right) \nonumber\\
		&   & \;\;\; \cdot \left( \left(\frac{A+B}{2}\right)^{2^{k-1}\cdot s}+ \left(\frac{A-B}{2}\right)^{2^{k-1}\cdot s} \right) \nonumber\\
		& =  &  H_{2^{k-1}\cdot s} \cdot G_{2^{k-1}\cdot s}.  \label{eq:HnHG}
	\end{IEEEeqnarray}
	The above induction completes the proof of the lemma.
\end{proof}

\begin{Proposition}
	The inverse of the $n$-th iteration of Cat map matrix
	\begin{equation}
	\textbf{C}^{-n}=
	\begin{bmatrix}
	\frac{1}{2}G_n+\frac{A-2}{2}H_n & -p\cdot H_n    \\
	-q\cdot H_n                      & \frac{1}{2}G_n-\frac{A-2}{2}H_n
	\end{bmatrix}.
	\label{eq:InverseiterateCat}
	\end{equation}
	\label{prop:inverse}
\end{Proposition}
\begin{proof}
Referring to Eq.~(\ref{eq:Gsm+1}) and Lemma~\ref{le:HG}, one has
	\begin{equation*}
	\begin{split}
	\textbf{C}^{2n} & =
	\begin{bmatrix}
	\frac{1}{2}G_{2n}-\frac{A-2}{2}H_{2n} & p\cdot H_{2n}    \\
	q\cdot H_{2n}                      & \frac{1}{2}G_{2n}+\frac{A-2}{2}H_{2n}
	\end{bmatrix}\\
	&=
	\begin{bmatrix}
	\frac{1}{2}G_{n}^2-1-\frac{A-2}{2}H_{n}G_n & p\cdot H_{n}G_n    \\
	q\cdot H_{n}G_n                      & \frac{1}{2}G_{n}^2-1+\frac{A-2}{2}H_{n}G_n
	\end{bmatrix}  \\
	&= G_n\cdot \textbf{C}^{n}+\begin{bmatrix}
	-1 & 0    \\
	0 & -1
	\end{bmatrix}.
	\end{split}
	\end{equation*}
	So, the $2n$-th iteration of Cat map matrix \eqref{eq:MatMatrix} satisfies
	\begin{IEEEeqnarray*}{rCl}
		G_n\cdot \textbf{C}^{n}-\textbf{C}^{2n} & = & \textbf{C}^{n}\cdot(G_n\cdot\mathcal{I}_2-\textbf{C}^{n})\\
		& = & \mathcal{I}_2.
	\end{IEEEeqnarray*}	
	Substituting $\textbf{C}^n$ in the above equation with Eq.~(\ref{eq:iterateCat}), one can get
	Eq.~(\ref{eq:InverseiterateCat}).	
\end{proof}

\subsection{Properties on iterating Cat map over $(\mathbb{Z}_{2^{\hat{e}}}, +,\ \cdot\ )$}

In this sub-section, how the graph structure of Cat map changes with the binary implementation precision are disclosed. To study the change process with the incremental increase of the precision $e$ from one, let $\hat{e}$ denote the given implementation precision instead, which is the upper bound of $e$.

Using Proposition~\ref{prop:leastPeriod} and Lemma~\ref{le:gcd} on the greatest common divisor
of three integers, the necessary and sufficient condition for the least period of Cat map~\eqref{eq:ArnoldInteger} over $(\mathbb{Z}_{2^e}, +,\ \cdot\ )$ is simplified as
Proposition~\ref{prop:periode}. Then, Lemmas~\ref{eq:Gscondition2l}, \ref{eq:Gscondition2le1}, and \ref{le:Hscondition} describe how the two parameters of the least period of Cat map, $\frac{1}{2}G_n$ and $H_n$, change
with the implementation precision.

\begin{Proposition}
The least period of Cat map~\eqref{eq:ArnoldInteger} over $(\mathbb{Z}_{2^e}, +,\ \cdot\ )$ is $T$ if and only if $T$ is the minimum value of $n$ satisfying	
\begin{empheq}[left=\empheqlbrace]{align}
	\frac{1}{2}	G_n & \equiv  1 \bmod 2^e,       \label{eq:Gne}	\\
	            H_n & \equiv  0 \bmod 2^{e-h_e}, \label{eq:Hneh}
\end{empheq}
where
\begin{equation}
h_e=
\begin{cases}
-1              & \mbox{if } e_p+e_q=0;\\
\min(e_p, e_q)  & \mbox{if } e>\min(e_p, e_q), e_p+e_q\neq 0;\\
e             & \mbox{if } e\le \min(e_p, e_q),
\end{cases}
\label{eq:he}
\end{equation}
$e_p=\max\{ x \mid p\equiv 0 \bmod 2^x \}$, $e_q=\max\{ x \mid q\equiv 0 \bmod 2^x \}$, $e\ge 1$, and $p, q\in \mathbb{Z}_{2^{\hat{e}}}$.
\label{prop:periode}
\end{Proposition}
\begin{proof}
Setting $N=2^e$ in Eq.~(\ref{eq:pqhGn2}), its last congruence becomes Eq.~(\ref{eq:Gne}).	Referring to Corollary~\ref{co:pqhGn2}, one can see that
this proposition can be proved by demonstrating that Eq.~(\ref{eq:Hneh}) is equivalent to the first three congruences in Eq.~(\ref{eq:pqhGn2}). Combining the first two congruences in Eq.~(\ref{eq:pqhGn2}), one has
\begin{equation}
H_n \equiv  0 \bmod 2^{e-h_{e, 1}},
\label{eq:Hnh1}
\end{equation}
where
$2^{e-h_{e, 1}}= \lcm\left( \frac{2^e}{\gcd(2^e, p)}, \frac{2^e}{\gcd(2^e, q)} \right)$.
Referring to Lemma~\ref{le:gcd},
one can get $h_{e, 1}=\max\{x\, | \gcd(p, q)\bmod 2^e\equiv 0 \bmod 2^x \}$
as
\[\lcm\left( \frac{2^e}{\gcd(2^e, p)}, \frac{2^e}{\gcd(2^e, q)} \right)=\frac{2^e}{\gcd( \gcd(2^e, p), \gcd(2^e, q))}.\]
The third congruence in Eq.~(\ref{eq:pqhGn2})
is equivalent to
\begin{equation}
H_n \equiv
\begin{cases}
0 \bmod \frac{2^e\cdot \gcd(2^e, 2)}{\gcd(2^e, p\cdot q)} & \mbox{if } e_p+e_q=0;\\
0 \bmod \frac{2^e}{\gcd(2^e, \frac{1}{2}p\cdot q)} & \mbox{if } e_p+e_q>0.
\end{cases}
\label{eq:Hnh2}
\end{equation}
Combining Eq.~(\ref{eq:Hnh1}) and Eq.~\eqref{eq:Hnh2}, one can get
$h_e=\min(h_{e, 1}, h_{e, 2})$ to assure that
Eq.~(\ref{eq:Hneh}) is equivalent to the first three congruences in Eq.~(\ref{eq:pqhGn2}),
where
\begin{equation}
h_{e, 2}=
\begin{cases}
\min(e, e_p+e_q)-1 & \mbox{if } e_p+e_q=0;\\
\min(e, e_p+e_q-1) & \mbox{if } e_p+e_q>0.
\end{cases}
\end{equation}
From the definition of $h_{e, 1}$ and $h_{e, 2}$, one has
\[
h_{e, 1}=
\begin{cases}
\min(e_p, e_q)  & \mbox{if } e>\min(e_p, e_q);\\
e               & \mbox{if } e\le \min(e_p, e_q),
\end{cases}
\]
and
\[
h_{e, 2}=
\begin{cases}
e_p+e_q-1       & \mbox{if } e\ge e_p+e_q;\\
e               & \mbox{if } e<e_p+e_q.
\end{cases}
\]
So,
\[
h_e=
\begin{cases}
\min(\min(e_p, e_q), e_p+e_q-1 ) & \mbox{if } e\ge e_p+e_q;\\
\min(\min(e_p, e_q), e)          & \mbox{if } \min(e_p, e_q)<e\\
& \hfill < e_p+e_q;\\
e                     & \mbox{if } e\le\min(e_p, e_q).
\end{cases}
\]
One can verify that
\[\min(\min(e_p, e_q), e_p+e_q-1 )=
\begin{cases}
\min(e_p, e_q) & \mbox{if } e_p+e_q>0;\\
-1             & \mbox{if } e_p+e_q=0.
\end{cases}
\]
If $\min(e_p, e_q)<e$, one has $\min(\min(e_p, e_q), e)=\min(e_p, e_q)$.
So, $h_e$ can be calculated as Eq.~(\ref{eq:he}).
\end{proof}

\begin{lemma}
For any integers $a, b, n$,
	one has
	\begin{equation*}
	\gcd\left(\gcd(d^n, a), \gcd(d^n, b)\right)=d^{n_g},
	\end{equation*}
	where \[n_g=\max\{ x\, |\, \gcd(a, b)\bmod d^n\equiv 0 \bmod d^x\},\]
	$d$ is a prime number, $\lcm$ and $\gcd$ denote the operator solving the least common multiple and greatest common divisor of two numbers, respectively.
\label{le:gcd}
\end{lemma}
\begin{proof}
	Let $a_d=\max\{ x\, |\, a \equiv 0 \bmod d^x\}$, $b_d=\max\{ x\, |\, b \equiv 0 \bmod d^x\}$,
	so 	$\min\{n, a_d,  b_d\}=\max\{ x\, |\, \gcd(a, b)\bmod d^n\equiv 0 \bmod d^x\}$.
	Then, one has	
	\begin{IEEEeqnarray*}{rCl}
		\IEEEeqnarraymulticol{3}{l}{\gcd\left(\gcd(d^n, a), \gcd(d^n, b)\right)}\nonumber\\\quad\quad
		&=&	 \gcd\left(\gcd(d^n, d^{a_d}), \gcd(d^n, d^{b_d})\right)\\
		&=&	d^{ \min\{ \min\{n, a_d\}, \min\{n, b_d\} \} }\\
		&=&	d^{ \min\{n, a_d,  b_d\} }\\
		&=& d^{n_g}.
	\end{IEEEeqnarray*}	
\end{proof}

\begin{lemma}
	Given an integer $e>1$, if $m$ and $s$ satisfy
	\begin{equation}
	\left\{
	\begin{split}
	\frac{1}{2}G_{2^m\cdot s} & \equiv 1 \bmod 2^{e},   \\
	\frac{1}{2}G_{2^m\cdot s} & \not\equiv   1 \bmod 2^{e+1},
	\end{split}
	\right.
	\label{eq:Gscondition}
	\end{equation}
	one has
	\begin{equation}
	\left\{
	\begin{split}	
	\frac{1}{2}G_{2^{m+l}\cdot s} &  \equiv 1 \bmod 2^{e+2l-1},  \\
	\frac{1}{2}G_{2^{m+l}\cdot s}     &  \equiv 1 \bmod 2^{e+2l},    \\
	\frac{1}{2}G_{2^{m+l}\cdot s}     &  \not\equiv 1 \bmod 2^{e+2l+1},
	\end{split}
	\right.
	\label{eq:Gscondition2l}
	\end{equation}	
	where $s$ and $l$ are positive integers and $m$ is a non-negative integer.
	\label{le:Gscondition}
\end{lemma}
\begin{proof}
	This lemma is proved via mathematical induction on $l$.
	From condition~(\ref{eq:Gscondition}), one can get $\frac{1}{2}G_{2^m\cdot s}=1+a_{e}\cdot 2^{e}$ , where $a_{e}$ is an odd integer. Then, one has
	\begin{IEEEeqnarray}{rCl}
		\frac{1}{2}G_{2^{m+1}\cdot s} & = &  \frac{1}{2}(G_{2^{m}\cdot s}^2-2) \nonumber \\
		& = &  \frac{1}{2}((2+a_e \cdot 2^{e+1})^2-2)   \nonumber \\
		& = &  2^{e+2}\cdot a_e\cdot (a_e\cdot 2^{e-1}+1)+1.
		\label{eq:G2msl=1}
	\end{IEEEeqnarray}
	As $a_e\cdot (a_e\cdot 2^{e-1}+1)$ is odd, condition~(\ref{eq:Gscondition2l}) exists for $l=1$. Assume that condition~(\ref{eq:Gscondition2l}) hold for $l=k$, namely
	$\frac{1}{2}G_{2^{m+k}\cdot s}=1+a_{e+2k}\cdot 2^{e+2k}$,
	where $a_{{e}+2k}$ is an odd integer.
	When $l=k+1$, one has
	\begin{IEEEeqnarray*}{rCl}
		\frac{1}{2}G_{2^{m+k+1}\cdot s}
		& = & \frac{1}{2}( G_{2^{m+k}\cdot s}^2-2)  \\
		& = & \frac{1}{2}((2+a_{e+2k} \cdot 2^{{e}+2k+1})^2-2)\\
		& = & a_{e+2k}^2\cdot 2^{2e+4k+1}+ a_{e+2k}\cdot 2^{e+2k+2} +1    \\
		& = & (a_{e+2k}^2\cdot 2^{e+2k-1}+ a_{e+2k})\cdot 2^{e+2k+2} +1.	
	\end{IEEEeqnarray*}
	As $a_{e+2k}^2\cdot 2^{e+2k-1}+ a_{e+2k}$ is an odd integer,
	condition~(\ref{eq:Gscondition2l}) also hold for $l=k+1$.
\end{proof}

\begin{lemma}
	If there is an odd integer $a_1$ satisfying $\frac{1}{2}G_{2^m\cdot s}=2\cdot a_1+1$, namely
	\begin{equation}
	\left\{
	\begin{split}
	\frac{1}{2}G_{2^m\cdot s} & \equiv 1 \bmod 2,   \\
	\frac{1}{2}G_{2^m\cdot s} & \not\equiv   1 \bmod 2^{2},
	\end{split}
	\right.
	\label{eq:Gsconditione1}
	\end{equation}
	one has
	\begin{equation}
	\left\{
	\begin{split}
	\frac{1}{2}G_{2^{m+1}\cdot s} &  \equiv 1 \bmod 2^{e_{g,0}}, \\
	\frac{1}{2}G_{2^{m+1}\cdot s} &  \not\equiv 1 \bmod 2^{e_{g,0}+1},
	\end{split}
	\right.
	\label{eq:Gscondition2le1}
	\end{equation}
	where
	$e_{g,0}=3+\max\{x \mid (a_1+1)\equiv 0 \bmod 2^x\}$.
	\label{le:Gms=1}
\end{lemma}
\begin{proof}
	From Eq.~(\ref{eq:G2msl=1}), one has
	$\frac{1}{2}G_{2^{m+1}\cdot s}=2^{3}\cdot a_1\cdot (a_1+1)+1.
	$ Then, condition~(\ref{eq:Gscondition2le1}) can be derived.
\end{proof}

\begin{lemma}
	If $G_s$ is even,
	\begin{equation}
	\left\{
	\begin{split}
	H_{2^m\cdot s} &  \equiv 0    \bmod 2^e,     \\
	H_{2^m\cdot s} & \not\equiv 0 \bmod 2^{e+1},
	\end{split}
	\right.
	\label{eq:Hscondition}
	\end{equation}
	one has
	\begin{empheq}[left=\empheqlbrace]{align}
	H_{2^{m+l}\cdot s} &  \equiv 0    \bmod 2^{e+l}, \label{eq:Hmscondition+l} \\
	H_{2^{m+l}\cdot s} & \not\equiv 0 \bmod 2^{e+l+1},\label{eq:Hmscondition+l+l}
	\end{empheq}
	where $l$ is a positive integer.
	\label{le:Hscondition}
\end{lemma}
\begin{proof}
	From Lemma~\ref{lemma:parity}, one has
	\begin{empheq}[left=\empheqlbrace]{align}
	\prod_{j=0}^{l-1} G_{2^j\cdot s} & \equiv 0    \bmod 2^l,  \label{eq:prodl} \\
	\prod_{j=0}^{l-1} G_{2^j\cdot s} & \not\equiv 0 \bmod 2^{l+1}, \label{eq:prodl+1}
	\end{empheq}
	since $2\mid G_{2^j\cdot s}$ and $4\nmid G_{2^j\cdot s}$ for any $j\in\{0, 2, \cdots, l-1\}$.
	From Lemma~\ref{le:HG}, one can get
	\begin{align}
	H_{2^{m+l}\cdot s} & = H_{2^{l}\cdot (2^{m}\cdot s)} \nonumber\\
	& = H_{2^{m}\cdot s} \cdot \prod_{j=0}^{l-1} G_{2^j\cdot s}.
	\end{align}
	So, Eq.~(\ref{eq:Hmscondition+l}) and inequality~(\ref{eq:Hmscondition+l+l})
	can be obtained by combining Eq.~(\ref{eq:Hscondition}) with
	Eq.~(\ref{eq:prodl}) and inequality~(\ref{eq:prodl+1}), respectively.
\end{proof}

\begin{Proposition}
For any $p, q$,
$G_{T_1}$ is even, and
\begin{equation}
\frac{1}{2}G_{T_1}-1=
\begin{cases}
\frac{1}{2}p\cdot q (p\cdot q+3)^2 & \mbox{if } p \mbox{ and } q \mbox{ are odd};\\
p\cdot q (\frac{1}{2}p\cdot q+2)   & \mbox{if } p \mbox{ or } q \mbox{ is odd}; \\
\frac{1}{2}p\cdot q                & \mbox{if } p \mbox{ and } q \mbox{ are even},
\end{cases}
\label{eq:gt1}
\end{equation}
where $T_1$ is the least period of Cat map~\eqref{eq:ArnoldInteger} over $(\mathbb{Z}_{2}, +, \cdot)$.	
\label{prop:Gt1even}
\end{Proposition}
\begin{proof}
Depending on the parity of $p$, $q$, the proof is divided into the following three cases:
\begin{itemize}
\item When $p, q$ are both odd:
\begin{equation*}
\textbf{C} \equiv
\begin{bmatrix}
1   & 0   \\
0   & 1
\end{bmatrix} \bmod 2.
\end{equation*}
One can calculate $T_1=3$. From Eq.~(\ref{eq:GH}), one has
\begin{IEEEeqnarray*}{rCl}
G_3 & = & \left(\frac{A+B}{2}\right)^3+\left(\frac{A-B}{2}\right)^3   \\
    & = & \frac{2A^3+6A\cdot B^2}{8} \\
    & = & A^3-3A  .
\end{IEEEeqnarray*}
As $A=p\cdot q+2$ is odd, $G_{T_1}$ is even.
\item When only $p$ or $q$ is odd:
$\textbf{C} \equiv
\begin{bsmallmatrix}
1   & 1   \\
0   & 1
\end{bsmallmatrix}\bmod 2$ if $p$ is odd; $\textbf{C} \equiv
\begin{bsmallmatrix}
1   & 0   \\
1   & 1
\end{bsmallmatrix}\bmod 2$ if $q$ is odd.
In either sub-case, $T_1=2$ and $A$ is even. As
\begin{IEEEeqnarray*}{rCl}
	G_2 & = & \left(\frac{A+B}{2}\right)^2+\left(\frac{A-B}{2}\right)^2 \\
	& = & A^2-2,
\end{IEEEeqnarray*}
one has $G_{T_1}$ is even.

\item When $p, q$ are both even: $T_1=1$.
So $G_{T_1}=A$ is also even.
\end{itemize}
Substituting $A=p\cdot q+2$ into $G_{T_1}$ in each above case, one can obtain Eq.~(\ref{eq:gt1}).
\end{proof}

\begin{Property}
The length of any cycle of Cat map~\eqref{eq:ArnoldInteger} implemented over $(\mathbb{Z}_{2^{\hat{e}}}, +, \cdot)$ comes from set
	$\{1 \} \cup \{2^k\cdot T_1\}_{k=0}^{\hat{e}-1}$, where
	\begin{equation*}
	T_1=
	\begin{cases}
	3  & \mbox{if } p \mbox{ and } q \mbox{ are odd};\\
	2  & \mbox{if only } p \mbox{ or } q \mbox{ is odd}; \\
	1  & \mbox{if } p \mbox{ and } q \mbox{ are even},
	\end{cases}
	\end{equation*}
	and $\hat{e}\ge 2$.
\label{prop:possibleLen}
\end{Property}
\begin{proof}
	As
	\begin{equation}
	\begin{bmatrix}
	1 & p\\
	q & 1+p\cdot q
	\end{bmatrix}
	\cdot
	\begin{bmatrix}
	x\\
	y
	\end{bmatrix}\bmod 2^{e}
	=
	\begin{bmatrix}
	x\\
	y
	\end{bmatrix}
	\label{eq:gnhn2epq}
	\end{equation}
	if and only if
	\begin{equation*}
	\left\{
	\begin{split}
	q\cdot x & \equiv 0\bmod 2^{e}, \\
	p\cdot y & \equiv 0\bmod 2^{e},
	\end{split}
	\right.
	\end{equation*}
	1 is the length of the cycles whose nodes satisfying condition~(\ref{eq:gnhn2epq}). From Fig.~\ref{fig:perioddistributione1}, one can see that
	\begin{equation}
	\begin{bmatrix}
	1 & p\\
	q & 1+p\cdot q
	\end{bmatrix}^{T_1}
	\cdot
	\begin{bmatrix}
	x\\
	y
	\end{bmatrix}\bmod 2
	=
	\begin{bmatrix}
	x\\
	y
	\end{bmatrix}
	\end{equation}
exists for any $x, y\in\mathbb{Z}_{2}$ if Eq.~(\ref{eq:gnhn2epq}) does not hold for $e=1$ and $T_1\neq 1$. From Fig.~\ref{Functionalgraphse1}, one has
	\begin{equation}
	\begin{bmatrix}
	1 & p\\
	q & 1+p\cdot q
	\end{bmatrix}^{n}
	\cdot
	\begin{bmatrix}
	x\\
	y
	\end{bmatrix}\bmod 2^2
	=
	\begin{bmatrix}
	x\\
	y
	\end{bmatrix}
	\label{eq:gnhn2e2}
	\end{equation}
	always holds for any $x, y\in\mathbb{Z}_{2^2}$ when $n=2\cdot T_1$ if Eq.~(\ref{eq:gnhn2e2}) does not hold for
	$n=T_1$.
	As $G_{2T_1}$ and $H_{2T_1}$ are both even,
	\[
	\begin{bmatrix}
	1 & p\\
	q & 1+p\cdot q
	\end{bmatrix}^{n}
	\cdot
	\begin{bmatrix}
	x \\
	y
	\end{bmatrix}\bmod 2^{e}
	=\begin{bmatrix}
	x \\
	y
	\end{bmatrix}
	\]
	can be presented as the equivalent form
	\begin{equation}
	\begin{bmatrix}
	x & 2\cdot p\cdot y-p\cdot q\cdot x\\
	y & 2\cdot q\cdot x+p\cdot q\cdot y
	\end{bmatrix}
	\cdot
	\begin{bmatrix}
	\frac{1}{2}G_n-1\\
	\frac{1}{2} H_n
	\end{bmatrix}\bmod 2^{e}
	=0
	\label{eq:gnhn2ee}
	\end{equation}
	when $n=2T_1$.
	If Eq.~(\ref{eq:gnhn2ee}) does not hold for $n=2T_1$ and $e\ge 3$,
	$n=2^2T_1$ is the least number of $n$ satisfying Eq.~(\ref{eq:gnhn2ee}) (See Lemmas~\ref{le:Gscondition} and \ref{le:Hscondition}). Referring to Lemmas~\ref{lemma:parity} and \ref{le:HG},
	$G_{2^mT_1}$ and $H_{2^mT_1}$ are both even for any positive integer
	$m$. So the equivalent form (\ref{eq:gnhn2ee}) can be reserved for any
	possible values of $n$.
	Iteratively repeat the above process, the length of cycle $n=2^k\cdot T_1$ can be obtained, where $k$ ranges from 0 to $\hat{e}-1$.
\end{proof}

If $p$ and $q$ are both odd, $H_{T_1}=(p\cdot q+1) \cdot (p\cdot q+3)\equiv 0\bmod 2^2$. As shown in Fig.~\ref{fig:SMNcat}, the length of the maximum cycle of Cat map over $\mathbb{Z}_{2^2}$ is $T_1$. So, the length of the maximum cycle is $3\cdot 2^{\hat{e}-2}$ if $\hat{e}\ge 3$. In addition, Property~\ref{prop:multiplecycle} is a direct consequence of
Property~\ref{prop:possibleLen}.

\begin{Proposition}
As for any $p$, $q$, $H_n$ is even,
where $n$ is the length of a cycle of Cat map~\eqref{eq:ArnoldInteger} larger than one.
\end{Proposition}
\begin{proof}
Referring to the definition of $H_{T_1}$ in Eq.~(\ref{eq:GH}),
\begin{equation}
H_{T_1}=
\begin{cases}
(p\cdot q+1) \cdot (p\cdot q+3) & \mbox{if } p \mbox{ and } q \mbox{ are odd};\\
p\cdot q+2                     & \mbox{if only } p \mbox{ or } q \mbox{ is odd}; \\
1                              & \mbox{if } p \mbox{ and } q \mbox{ are even}
\end{cases}
\label{eq:ht}
\end{equation}
can be calculated as the proof of Proposition~\ref{prop:Gt1even}.
\begin{itemize}
\item When $p$ or $q$ is even:
$H_2=H_1\cdot G_1=p\cdot q+2$ is even.

\item When $p$ and $q$ are both odd:
$H_3=(p\cdot q+1) \cdot (p\cdot q+3)$ is even.
\end{itemize}
Referring to Property~\ref{prop:possibleLen}, $H_n=H_{2^m\cdot s}$,
where
\begin{equation*}
s=
\begin{cases}
2   & \mbox{if } p \mbox{ and } q \mbox{ are even};\\
T_1 & \mbox{otherwise}.
\end{cases}
\end{equation*}
From Lemma~\ref{le:HG}, one can get $H_n$ is even.
\end{proof}

\subsection{Disclosing the regular graph structure of Cat map}

With increase of $e$, $\frac{1}{2}G_{2^nT_1}$ and $H_{2^nT_1}$ will reach
the balancing condition given in Proposition~\ref{prop:periode}. As shown the proof in Theorem~\ref{the:threshold}, the explicit presentation of the threshold value of $e$, $e_s$, is obtained. When $e\ge e_s$, the period
of Cat map double for every increase of $e$ by one. As for Table~\ref{table:pqeTc}, $e_s=8$ (The dashlined row with tiny gap).

\begin{theorem}
There exists a threshold value of $e$, $e_{s}$, satisfying
\begin{equation}
T_{e+l}=2^l \cdot T_{e}
\end{equation}
when $e\ge e_s$, where $T_e$ is the period of Cat map over $(\mathbb{Z}_{2^e}, +,\ \cdot\ )$, and $l$ is a non-negative integer.
\label{the:threshold}
\end{theorem}
\begin{proof}
When $e=1$, one has
\begin{empheq}[left=\empheqlbrace]{align}
	\frac{1}{2}	G_{T_1} & \equiv  1 \bmod 2  \nonumber	\\
		H_{T_1}         & \equiv  0 \bmod 2^{1-h_1}, \label{eq:HT1}
\end{empheq}
from Proposition~\ref{prop:periode}.
From Eq.~(\ref{eq:HT1}), one can get $e_{s, h}$, the minimal number of $e$ satisfying
\begin{equation*}
\left\{
\begin{split}
H_{T_1} &  \equiv 0    \bmod 2^{e},   \\
H_{T_1} & \not\equiv 0 \bmod 2^{e+1}, \\
e       & \ge  1-h_1.
\end{split}
\right.
\label{eq:Hscondition4}
\end{equation*}
From Lemma~\ref{le:Gms=1}, one can get the minimum positive number of $e$ satisfying
\begin{equation}
\left\{
\begin{split}
\frac{1}{2}G_{2^{m_0}\cdot T_1} & \equiv 1 \bmod 2^{e},     \\
\frac{1}{2}G_{2^{m_0}\cdot T_1} & \not\equiv  1 \bmod 2^{e+1},
\end{split}
\right.
\label{eq:escondition2}
\end{equation}
$e_{s, g}$, by increasing $e$ from 1, where
\begin{equation*}
m_0=
\begin{cases}
1  & \mbox{if } \frac{1}{2}\cdot G_{T_1} \not\equiv  1 \bmod 2^{2};\\
0  & \mbox{otherwise.}
\end{cases}
\end{equation*}

\setlength\tabcolsep{4pt} 
\addtolength{\abovecaptionskip}{0pt}
\begin{table*}[!htb]
	\centering  
	\caption{The threshold values $e_{s}$, $e_{s}'$ under various combinations of $(p,q)$.}
	\begin{tabular}{*{16}{c|}c} 
		\hline 
		\diagbox[width=6em]{$p$}{$e_s (e'_{s})$}{$q$}
		&$1$&   $2$&   $3$&   $4$&   $5$&   $6$&   $7$&   $8$&   $9$&   $10$&   $11$&   $12$&   $13$&   $14$&   $15$   &$16$\\ \hline
		1&  2(2)& 3(3)& 3(3)& 1(1)& 3(3)& 4(4)& 4(4)& 1(1)& 2(2)& 3(3)& 3(3)& 1(1)& 4(4)& 5(5)& 5(5)&1(1)\\
		2&  3(3)& 2(1)& 4(4)& 1(1)& 3(3)& 2(1)& 5(5)& 1(1)& 3(3)& 2(1)& 4(4)& 1(1)& 3(3)& 2(1)& 6(6)&1(1)\\
		3&  3(3)& 4(4)& 2(2)& 1(1)& 5(5)& 3(3)& 3(3)& 1(1)& 3(3)& 6(6)& 2(2)& 1(1)& 4(4)& 3(3)& 4(4)&1(1)\\
		4&  1(1)& 1(1)& 1(1)& 2(2)& 1(1)& 1(1)& 1(1)& 2(2)& 1(1)& 1(1)& 1(1)& 2(2)& 1(1)& 1(1)& 1(1)&2(2)\\
		5&  3(3)& 3(3)& 5(5)& 1(1)& 2(2)& 6(6)& 3(3)& 1(1)& 4(4)& 3(3)& 4(4)& 1(1)& 2(2)& 4(4)& 3(3)&1(1)\\
		6&  4(4)& 2(1)& 3(3)& 1(1)& 6(6)& 2(1)& 3(3)& 2(1)& 4(4)& 2(1)& 3(3)& 1(1)& 5(5)& 2(1)& 3(3)&1(1)\\
		7&  4(4)& 5(5)& 3(3)& 1(1)& 3(3)& 3(3)& 2(2)& 1(1)& 7(7)& 4(4)& 4(4)& 1(1)& 3(3)& 3(3)& 2(2)&1(1)\\
		8&  1(1)& 1(1)& 1(1)& 2(2)& 1(1)& 1(1)& 1(1)& 3(3)& 1(1)& 1(1)& 1(1)& 2(2)& 1(1)& 1(1)& 1(1)&3(3)\\
		9&  2(2)& 3(3)& 3(3)& 1(1)& 4(4)& 4(4)& 7(7)& 1(1)& 2(2)& 3(3)& 3(3)& 1(1)& 3(3)& 8(8)& 4(4)&1(1)\\
		10& 3(3)& 2(1)& 6(6)& 1(1)& 3(3)& 2(1)& 4(4)& 1(1)& 3(3)& 2(1)& 5(5)& 1(1)& 3(3)& 2(1)& 4(4)&1(1)\\
		11& 3(3)& 4(4)& 2(2)& 1(1)& 4(4)& 3(3)& 4(4)& 1(1)& 3(3)& 5(5)& 2(2)& 1(1)& 5(5)& 3(3)& 3(3)&1(1)\\
		12& 1(1)& 1(1)& 1(1)& 2(2)& 1(1)& 1(1)& 1(1)& 2(2)& 1(1)& 1(1)& 1(1)& 2(2)& 1(1)& 1(1)& 1(1)&2(2)\\
		13& 4(4)& 3(3)& 4(4)& 1(1)& 2(2)& 5(5)& 3(3)& 1(1)& 3(3)& 3(3)& 5(5)& 1(1)& 2(2)& 4(4)& 3(3)&1(1)\\
		14& 5(5)& 2(1)& 3(3)& 1(1)& 4(4)& 2(1)& 3(3)& 1(1)& 8(8)& 2(1)& 3(3)& 1(1)& 4(4)& 2(1)& 3(3)&1(1)\\
		15& 5(5)& 6(6)& 4(4)& 1(1)& 3(3)& 3(3)& 2(2)& 1(1)& 4(4)& 4(4)& 3(3)& 1(1)& 3(3)& 3(3)& 2(2)&1(1)\\
		16& 1(1)& 1(1)& 1(1)& 2(2)& 1(1)& 1(1)& 1(1)& 3(3)& 1(1)& 1(1)& 1(1)& 2(2)& 1(1)& 1(1)&	1(1)&4(4)\\
		\hline 
	\end{tabular}
	\label{table:nume_vs}
\end{table*}	

Referring to Proposition~\ref{prop:Gt1even}, $G_{T_1}$ is even.
So, one can get
\begin{empheq}[left=\empheqlbrace]{align*}
\frac{1}{2}	G_{2^{m_0+x}\cdot T_1} & \equiv  1 \bmod 2^{e_{s, g}+2\cdot x}\\
H_{2^{m_0+x}\cdot T_1}             & \equiv  0 \bmod 2^{m_0+x+e_{s, h}}
\end{empheq}
by referring to Lemmas \ref{le:Hscondition} and ~\ref{le:Gscondition},
where $x$ is a non-negative integer. Note that
$h_e$ is monotonically increasing with respect to $e$ and fixed as $\hat{h}_{e}$ when $e\ge \min(e_p, e_q)$ (See Eq.~(\ref{eq:he})),
where
\[
\hat{h}_e=
\begin{cases}
-1               & \mbox{if } e_p+e_q=0;\\
\min(e_p, e_q)   & \mbox{otherwise}.
\end{cases}
\]
Set
\begin{equation}
e_s=(e_{s,h}+m_0+\hat{h}_e) +x_0,
\label{eq:es}
\end{equation}
one has
$e_s\ge \min(e_p, e_q)$ from Eq.~(\ref{eq:ht}),
\begin{empheq}[left=\empheqlbrace]{align*}
\frac{1}{2}	G_{2^{m_0+x_0}\cdot T_1} & \equiv  1 \bmod 2^{e_{s, g}+2x_0} \\
\frac{1}{2}	G_{2^{m_0+x_0}\cdot T_1} & \not\equiv  1 \bmod 2^{e_{s, g}+2x_0+1}
\end{empheq}
and
\begin{empheq}[left=\empheqlbrace]{align*}
	H_{2^{m_0+x_0}\cdot T_1}  & \equiv  0 \bmod 2^{e_s-\hat{h}_e}  \\
    H_{2^{m_0+x_0}\cdot T_1}  & \not\equiv  0 \bmod 2^{e_s+1-\hat{h}_{e}}
\end{empheq}
where
\begin{equation}
x_0=
\begin{cases}
(e_{s, h}+m_0+\hat{h}_e)-e_{s, g} & \mbox{if } e_{s, g}<e_{s, h}+m_0+\hat{h}_e;\\
0  & \mbox{otherwise.}
\end{cases}
\label{eq:x0Condition}
\end{equation}
Referring to Lemma~\ref{le:Gscondition}, one has
\begin{equation}
\frac{1}{2}	G_{2^{m_0+x_0+l}\cdot T_1} \equiv  1 \bmod 2^{e_{s, g}+2x_0+2l}.
\label{eq:es2l}
\end{equation}
Combing Lemma~\ref{le:Hscondition},
\begin{empheq}[left=\empheqlbrace]{align}
\frac{1}{2}	G_{2^{m_0+x_0+l}\cdot T_1} & \equiv  1 \bmod 2^{e_s+l}  \nonumber\\
H_{2^{m_0+x_0+l}\cdot T_1}             & \equiv  0 \bmod 2^{e_s+l-\hat{h}_e}\label{eq:hscondition3}
\end{empheq}
as $e_{s, g}+2x_0+2l\ge e_s+l$ for any non-negative integer $l$. Hence, by Proposition~\ref{prop:periode},
\begin{equation*}
T_{e+l}=2^l \cdot T_{e}=2^{m_0+x_0+l}\cdot T_1
\end{equation*}
when $e\ge e_s$.
\end{proof}

From the proof of Theorem~\ref{the:threshold}, one can see that the value of $e_s$
in Eq.~(\ref{eq:es}) is conservatively estimated to satisfy the required conditions
(The balancing conditions may be obtained when $h_e$ is still not approach $\hat{h}_e$).
In practice, there exists another real threshold value of $e$, $e'_{s}\le e_{s}$, satisfying
\begin{equation*}
T_{e'_s+l}=2^l \cdot T_{e'_s},
\end{equation*}
which is verified by Table~\ref{table:nume_vs}.

\begin{figure*}[!htb]
	\centering
	\begin{minipage}{1.8\twofigwidth}
		\centering
		\includegraphics[width=1.8\twofigwidth]{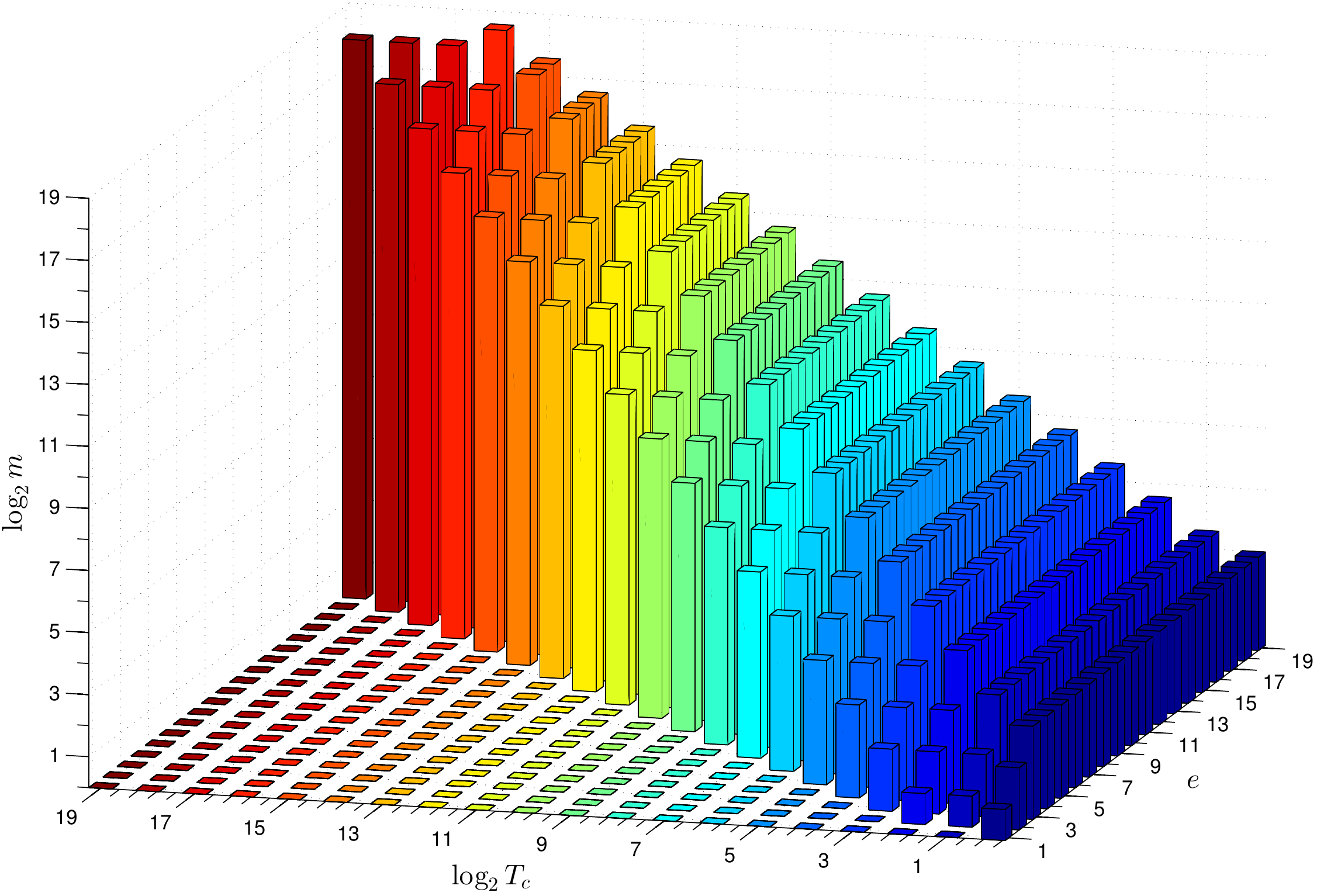}
		a)
	\end{minipage}
	\begin{minipage}{1.8\twofigwidth}
		\centering
		\includegraphics[width=1.8\twofigwidth]{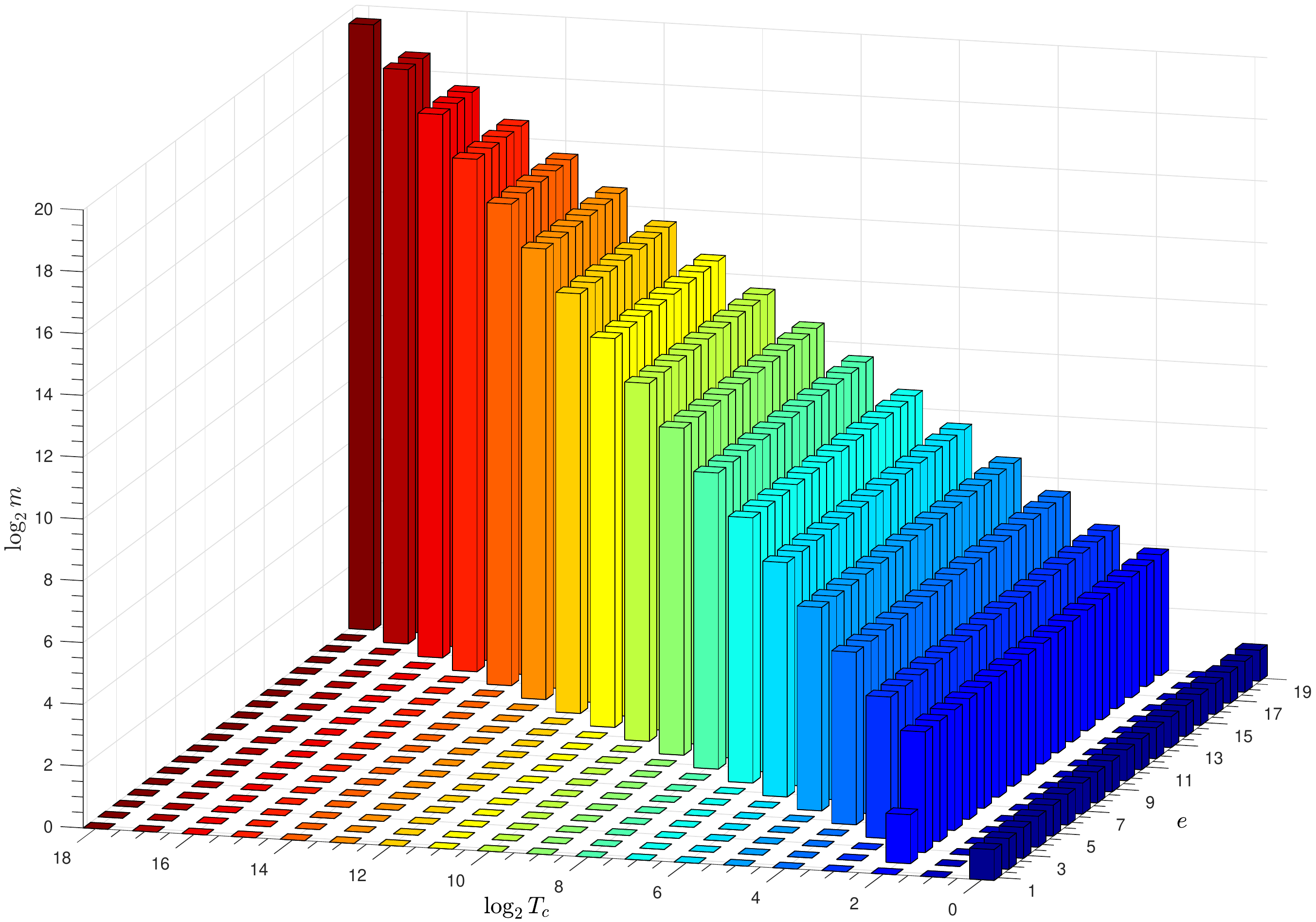}
		b)
	\end{minipage}\\
	\begin{minipage}{1.8\twofigwidth}
		\centering
		\includegraphics[width=1.8\twofigwidth]{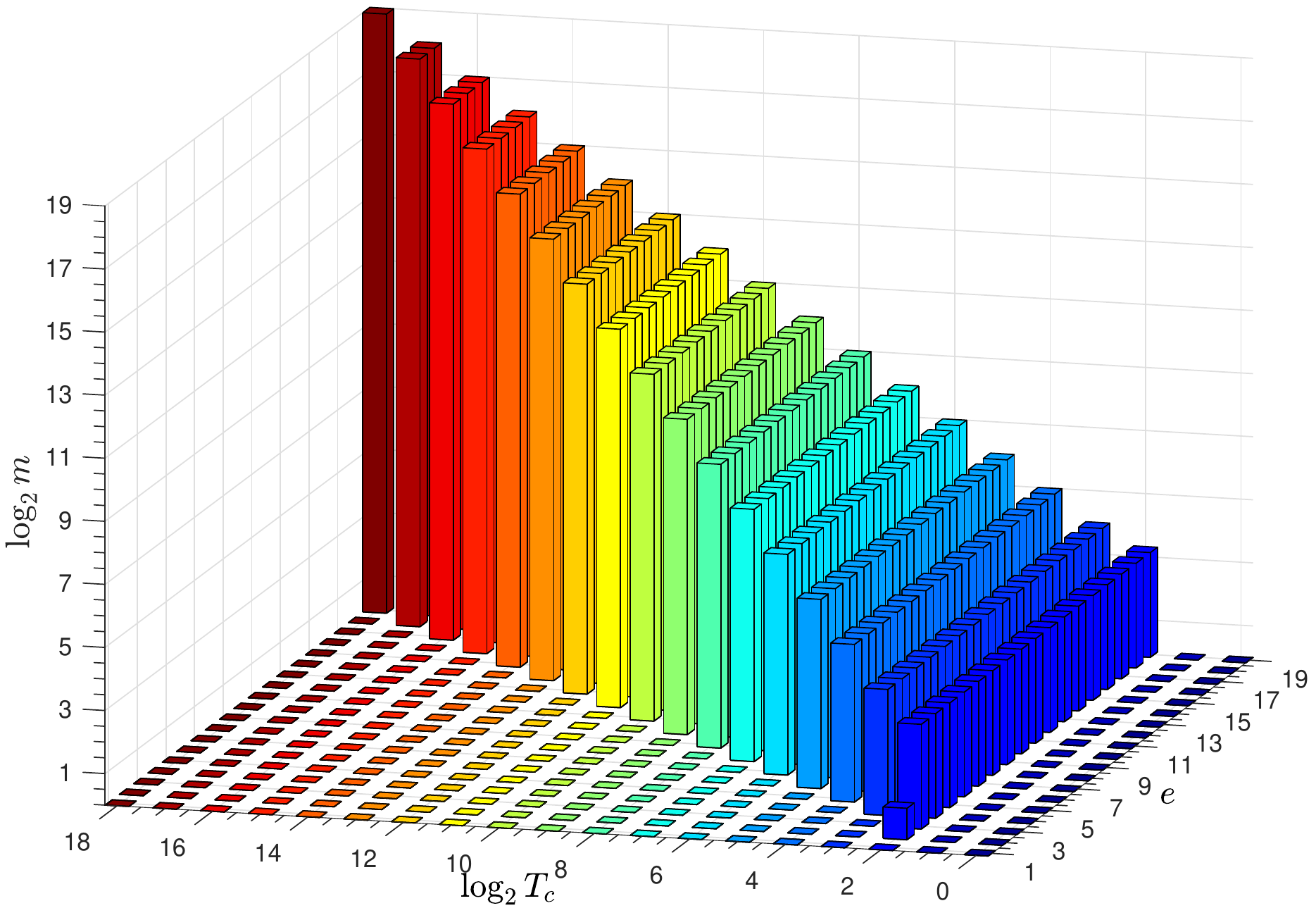}
		c)
	\end{minipage}
	\begin{minipage}{1.8\twofigwidth}
		\centering
		\includegraphics[width=1.8\twofigwidth]{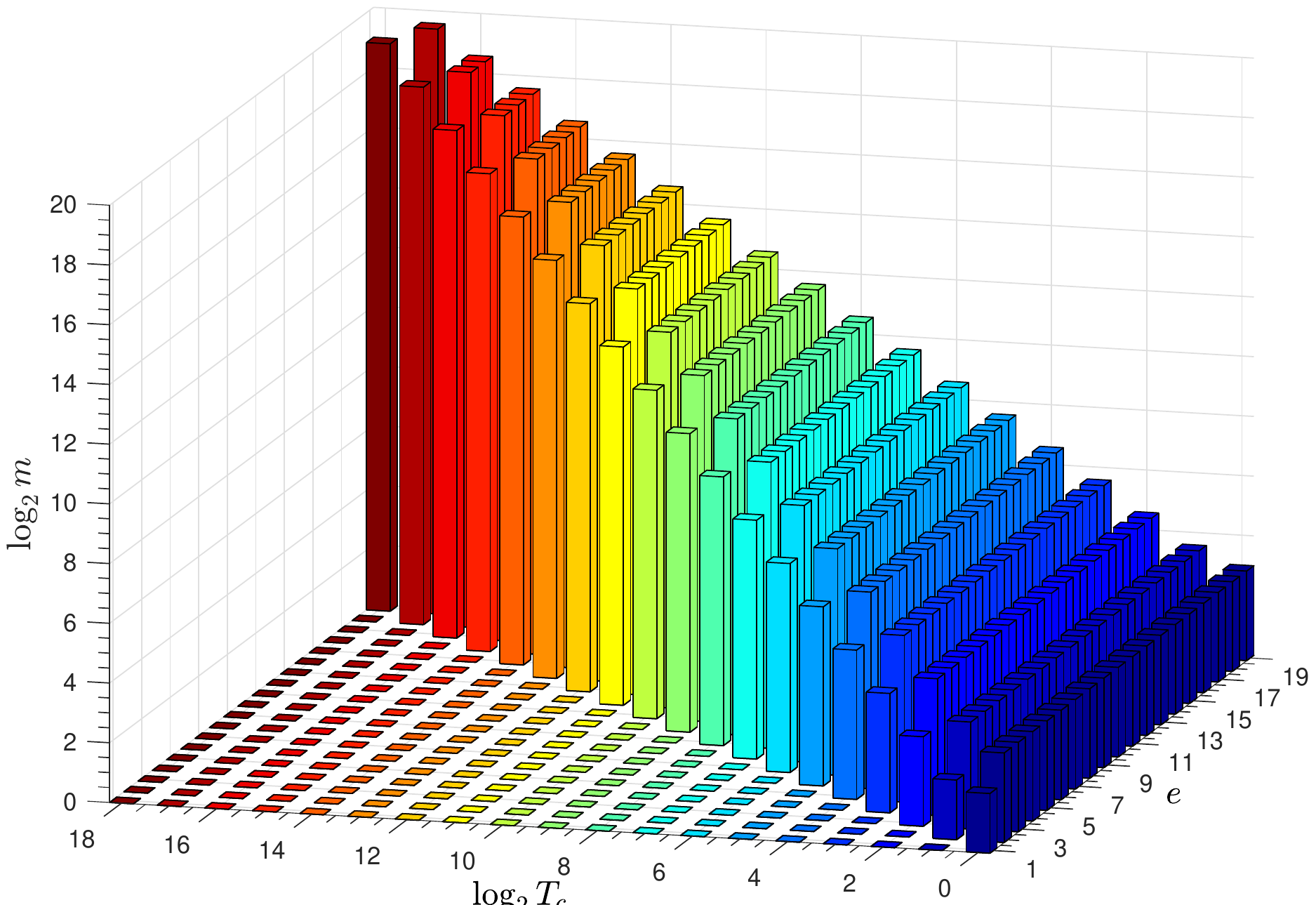}
		d)
	\end{minipage}
	\caption{The cycle distribution of Cat map~(\ref{eq:ArnoldInteger}) over $\mathbb{Z}_{2^e}$, $e=1\sim 19$: a) $(p, q)=(7, 8)$; b) $(p, q)=(6, 7)$; c) $(p, q)=(5, 7)$; d) $(p, q)=(12, 14)$.}
	\label{fig:perioddistributionab}
\end{figure*}

\begin{theorem}
When $T_c>T_{e_s}$,
\begin{equation}
2N_{T_c, e}=N_{2T_c, e+1},
\end{equation}
where $N_{T_c, e}$ is the number of cycles with period $T_c$ of Cat map~\eqref{eq:ArnoldInteger} over $(\mathbb{Z}_{2^{e}}, +, \cdot)$.
\label{the:numbercycle}
\end{theorem}
\begin{proof}
As for any point $(x, y)$ of a cycle with the least period $T_c$ in $F_e$, one has
\begin{empheq}[left=\empheqlbrace]{align}
\textbf{C}^{T_c}\cdot
\begin{bmatrix}
x \\
y
\end{bmatrix}\bmod 2^{e}
&=\begin{bmatrix}
x \\
y
\end{bmatrix}\label{eq:Tc}  \\
\textbf{C}^{\frac{T_c}{2}}\cdot
\begin{bmatrix}
x \\
y
\end{bmatrix}\bmod 2^{e}
&\neq\begin{bmatrix}
x \\
y
\end{bmatrix}
\label{ineq:Tc}
\end{empheq}
by referring to Property~\ref{prop:multiplecycle} and Table~\ref{table:num2}.
Referring to Eq.~(\ref{eq:es2l}) and Eq.~(\ref{eq:hscondition3}), one has
\begin{equation}
\left\{
\begin{split}
e_{g, n} & =e_{s,g}+2x_0+2l, \\
e_{h, n} & =e_s-\hat{h}_e+l,
\end{split}
\right.
\end{equation}
where $e_{g, n}=\max\{x\mid \frac{1}{2}G_n \equiv 1 \bmod 2^x \}$,
$e_{h, n}=\max\{x\mid H_n\equiv 0 \bmod 2^x \}$,
$n=2^{l}\cdot T_{e_s}$, $l$ is a non-negative integer.
So, one can get
\begin{multline}
\begin{bmatrix}
	 x & 2\cdot p\cdot y-p\cdot q\cdot x\\
	 y & q\cdot x+\frac{1}{2}p\cdot q\cdot y
\end{bmatrix}
\cdot
\begin{bmatrix}
	\frac{1}{2}G_{2n}-1\\
H_{2n}
\end{bmatrix}\bmod 2^{e+1}
=0
\label{condi:g2nh2n}
\end{multline}
from Eq.~(\ref{eq:gnhn2ee}).
Setting $n=T_c$,
\begin{equation}
\textbf{C}^{2T_c} \cdot
\begin{bmatrix}
	x\\
	y
\end{bmatrix}\bmod 2^{e+1}=
\begin{bmatrix}
	x \\
	y
\end{bmatrix}.
\label{eq:Gtc2Tc}
\end{equation}
Referring to Lemma~\ref{prop:inverse}, and $G_n$ is even,
\begin{IEEEeqnarray}{rCl}
	\IEEEeqnarraymulticol{3}{l}{\textbf{C}^{2n}\cdot
		\begin{bmatrix}
			x\\
			y
		\end{bmatrix}\bmod 2^{e+1}}\nonumber\\* \quad
	& = & (G_{n}\cdot \textbf{C}^{n}-\mathcal{I}_2)
	\begin{bmatrix}
		x\\
		y
	\end{bmatrix}\bmod 2^{e+1}\nonumber  \\
	& = &
	\left (G_{n} \textbf{C}^{n}\begin{bmatrix}
		x\\
		y
	\end{bmatrix}-
	\begin{bmatrix}
		x \\
		y
	\end{bmatrix}\right)\bmod 2^{e+1} \label{eq:C2n}.
\end{IEEEeqnarray}
Therefore,
\begin{IEEEeqnarray}{rCl}
	\IEEEeqnarraymulticol{3}{l}{\textbf{C}^{2n}\cdot
		\begin{bmatrix}
			a\cdot 2^e\\
			b\cdot 2^e
		\end{bmatrix}\bmod 2^{e+1}}\nonumber\\* \quad
	&=&
	\left (G_{n} \textbf{C}^{n}\begin{bmatrix}
	a\cdot 2^e\\
b\cdot 2^e
\end{bmatrix}-
\begin{bmatrix}
	a\cdot 2^e \\
	b\cdot 2^e
\end{bmatrix}\right)\bmod 2^{e+1} \nonumber\\
   &=&
   	\begin{bmatrix}
   	a\cdot 2^e\\
   	b\cdot 2^e
   \end{bmatrix}\bmod 2^{e+1}, \label{eq:C2nab}
\end{IEEEeqnarray}
where $a, b\in \{0, 1\}$. Combing Eq.~(\ref{eq:Gtc2Tc}) and Eq.~(\ref{eq:C2nab}) with $n=T_c$,
one can get
	\[\textbf{C}^{2T_c}\cdot \begin{bmatrix}
			x+a\cdot 2^e\\
			y+b\cdot 2^e
		\end{bmatrix}\bmod 2^{e+1}
	=
	\begin{bmatrix}
		x+a\cdot 2^e\\
		y+b\cdot 2^e
	\end{bmatrix}.\]
	
Setting $n=\frac{T_c}{2}$ in the left-hand side of Eq.~(\ref{condi:g2nh2n}), one has
\[\textbf{C}^{T_c}\cdot \begin{bmatrix}
x\\
y
\end{bmatrix}\bmod 2^{e+1}
\neq
\begin{bmatrix}
x\\
y
\end{bmatrix}
\]
from Eq.~(\ref{ineq:Tc}) as $\frac{T_c}{2}\ge T_{e_s}$. Combing the above inequalities with Eq.~(\ref{eq:C2nab}), one has
\[
\textbf{C}^{T_c}\cdot \begin{bmatrix}
x+a\cdot 2^e\\
y+b\cdot 2^e
\end{bmatrix}\bmod 2^{e+1}
\neq
\begin{bmatrix}
x+a\cdot 2^e\\
y+b\cdot 2^e
\end{bmatrix}.
\]
So, $2T_c$ is the least period of $(x+a\cdot 2^e, y+b\cdot 2^e)$ in $F_{e+1}$ for any $a, b\in\{0, 1\}$. When $T_c>T_{e_s}$, from Property~\ref{Prop:cycleExpan}, one has
\[
4\cdot N_{T_c,e}\cdot T_c=N_{2T_c,e+1}\cdot 2\cdot T_c,
\]
namely
$2N_{T_c,e}=N_{2T_c,e+1}$.
\end{proof}

\begin{lemma}
As for any point $(x, y) $ in a cycle of length $n$ of Cat map~\eqref{eq:ArnoldInteger} over $(\mathbb{Z}_{2^{e}}, +, \cdot)$,
\begin{equation}
(G_n-2)\cdot
	\begin{bmatrix}
		x \\
		y
	\end{bmatrix}
   \bmod 2^{e}=0,
\label{eq:gn2e}
\end{equation}
where $n>1$.
\label{le:gn2}
\end{lemma}
\begin{proof}
Substituting
\begin{equation*}
\begin{split}
\begin{bmatrix}
\frac{1}{2}G_{2n}-1\\
H_{2n}
\end{bmatrix}
& =
\begin{bmatrix}
\frac{1}{2}G_{n}^2-2\\
 G_n H_{n}
\end{bmatrix}\\
&=
G_n
\begin{bmatrix}
\frac{1}{2}G_{n}-1   \\
H_{n}
\end{bmatrix}+
\begin{bmatrix}
G_n-2 \\
0
\end{bmatrix}
\end{split}
\end{equation*}
into Eq.~(\ref{condi:g2nh2n}), one can obtain
Eq.~(\ref{eq:gn2e}).
\end{proof}

\begin{lemma}
When $e>e_0$, any point $(x, y) $ in a cycle of length $T_c$ of Cat map~\eqref{eq:ArnoldInteger} over $(\mathbb{Z}_{2^{e}}, +, \cdot)$ satisfies
\begin{equation}
\begin{bmatrix}
x \\
y
\end{bmatrix}
\bmod 2=0,
\label{eq:2x2y}
\end{equation}
where
	\begin{equation*}
	e_0=
	\begin{cases}
	\max(e_p, e_q)        & \mbox{if } T_c=1; \\
	e_{s, g}+1            & \mbox{if } l_c=0, T_1\neq 1;\\
	e_{s, g}+2\cdot l_c-1 & \mbox{if } 1\leq l_c \leq s+1,
	\end{cases}
	\end{equation*}
$2^{l_c}\cdot T_1=T_c$, and $2^s\cdot T_1=T_{e_s}$.
\label{le:e1}
\end{lemma}
\begin{proof}
When $T_c=1$ and $e\ge\max(e_p, e_q)+1$, condition~(\ref{eq:2x2y}) should exist to satisfy Eq.~(\ref{eq:gnhn2epq}).
Setting $m_0=0$ in Eq.~(\ref{eq:escondition2}), one has
\begin{equation*}
\left\{
\begin{split}
\frac{1}{2}G_{T_1} & \equiv 1 \bmod 2^{e_1},     \\
\frac{1}{2}G_{T_1} & \not\equiv  1 \bmod 2^{e_1+1},
\end{split}
\right.
\end{equation*}
where
\begin{equation*}
e_1=
\begin{cases}
e_{s, g}         & \mbox{if } \frac{1}{2}G_{T_1}\equiv 1 \bmod 2^2;\\
1 & \mbox{otherwise},
\end{cases}
\end{equation*}
Referring to Lemma~\ref{le:gn2}, if $T_1\neq 1$ and $e\ge e_{s, g}+2\ge e_1+2$,
any point of a cycle of length $T_1$ should satisfy condition~(\ref{eq:2x2y}) to
meet Eq.~(\ref{eq:gn2e}).
Setting $m_0=1$ in Eq.~(\ref{eq:escondition2}), one has
\begin{equation*}
\left\{
\begin{split}
\frac{1}{2}G_{2T_1} & \equiv 1 \bmod 2^{e_{s,g}},     \\
\frac{1}{2}G_{2T_1} & \not\equiv  1 \bmod 2^{e_{s, g}+1}.
\end{split}
\right.
\end{equation*}
Referring to Lemma~\ref{le:Gscondition}, one can further get
\begin{equation*}
\left\{
\begin{split}
\frac{1}{2}G_{2^{l_c+1}T_1} & \equiv 1 \bmod 2^{e_{s,g}+2l_c}   \\
\frac{1}{2}G_{2^{l_c+1}T_1} & \not\equiv  1 \bmod 2^{e_{s, g}+2l_c+1}
\end{split}
\right.
\end{equation*}
for $l_c=1\sim s+1$.
Referring to Lemma~\ref{le:gn2}, if $e\ge e_{s, g}+2\cdot l_c$,
any point of a cycle of length $2^{l_c}\cdot T_1$ should satisfy condition~(\ref{eq:2x2y}) to
meet Eq.~(\ref{eq:gn2e}).
\end{proof}

From Theorem~\ref{the:numbercycle}, one can see that the number of cycles of various lengths in $F_e$ can be easily deduced from that of $F_{e_s}$ when $e>e_s$. As for any cycle with length $T_c\leq T_{e_s}$, the threshold values of $e$ in condition~(\ref{cond:eL+1}) are given to satisfy Eq.~(\ref{prop:Ntce2}). As for the cycles with length $T_c>T_{e_s}$,
the threshold values can be directly calculated with Theorem~\ref{the:numbercycle}.
As shown in
Theorem~\ref{the:NTc}, as for every possible length of cycle, the number of the cycles
of the length become a fixed number when $e$ is sufficiently large.
The strong regular graph patterns demonstrated in Table~\ref{table:pqeTc} and Fig.~\ref{fig:perioddistributionab} are rigorously proved in Theorems~\ref{the:numbercycle}, \ref{the:NTc}. Now, we can see that
the exponent value of the distribution function of cycle lengths of $F_e$ is fixed two when $e$ is sufficiently large.

\begin{theorem}
When
\begin{equation}
e\ge
\begin{cases}
\max(e_p, e_q)        & \mbox{if } T_c=1; \\
e_{s, g}+1            & \mbox{if } l_c=0, T_1\neq 1;\\
e_{s, g}+2\cdot l_c-1 & \mbox{if } 1\leq l_c \leq s+1;\\
e_{s, g}+s+1+l_c      & \mbox{if } l_c \ge s+2,
\end{cases}
\label{cond:eL+1}
\end{equation}
one has
	\begin{equation}
	N_{T_c, e}=N_{T_c, e+l},
	\label{prop:Ntce2}
	\end{equation}
where
$2^{l_c}\cdot T_1=T_c$, $2^s\cdot T_1=T_{e_s}$, and
$l$ is any positive integer.
\label{the:NTc}
\end{theorem}
\begin{proof}
From Property~\ref{prop:isomorphism}, one can conclude that the number of cycles of length
	$T_c$ in $F_{e+1}$ is larger than or equal to that in $F_{e+1}$, i.e.
	$N_{T_c, e}\le N_{T_c, e+1}$ for any $e$.
	
Referring to Lemma~\ref{le:e1}, as for any point $(x, y)$ in a cycle of length $T_c=2^{l_c}\cdot T_1$,
\begin{equation*}
\begin{bmatrix}
\frac{x}{2} & 2\cdot p\cdot \frac{y}{2}-p\cdot q\cdot \frac{x}{2}\\
\frac{y}{2} & 2\cdot q\cdot \frac{x}{2}+p\cdot q\cdot \frac{y}{2}
\end{bmatrix}
\cdot
\begin{bmatrix}
\frac{1}{2}G_{T_c}-1\\
\frac{1}{2} H_{T_c}
\end{bmatrix}\bmod 2^{e-1}
=0
\end{equation*}
if $e$ satisfy condition~(\ref{cond:eL+1}), meaning that
$N_{T_c, e}\ge N_{T_c, e+1}$.
So $N_{T_c, e}=N_{T_c, e+1}$.
$N_{2T_{e_s}, e}=N_{2T_{e_s}, e+l}$ for any $l$.
From Theorem~\ref{the:numbercycle}, when $l_c\ge s+2$,
$e\ge e_{s, g}+2s+2+l_c-s-1=e_{s, g}+s+1+l_c$,
\begin{equation*}
N_{T_c, e}=2^{l_c-s-1}\cdot N_{T_{e_s}, e^\star}=N_{T_c, e+l}.
\end{equation*}
\end{proof}

\section{Application of the cycle structure of Cat map}
\label{sec:apply}

In this section, we briefly discuss application of the obtained results on the cycle structure of Cat map
in theoretical analysis and cryptographical application.


The infinite number of unstable periodic orbits (UPO's) of a chaotic system constitute its skeleton \cite{Davidchack1999PRE}. In a finite-precision domain, any periodic orbit is a cycle (See Property~\ref{prop:onecycle}). Its stability is dependent on change trend of its distance with the neighboring states in the phase space. To show the relative positions among different cycles, we depicted real image of the SMN shown in Fig.~\ref{fig:SMNcat}c) in Fig.~\ref{fig:upo}. Note that some states may be located in different cycles due to the finite-precision effect. The main result obtained in \cite{Catchen2013period2e}, i.e. Eq.~(\ref{eq:numberMaps}), describes the number of different possible Cat maps owning a specific period, which has no any relationship with the number of cycles and distance between a cycle and its neighboring states. So, it cannot help to identify unstable periodic orbits of the original chaotic Cat map at all.

\begin{figure}[!htb]
	\centering
	\begin{minipage}{1.5\twofigwidth}
		\centering
		\includegraphics[width=1.5\twofigwidth]{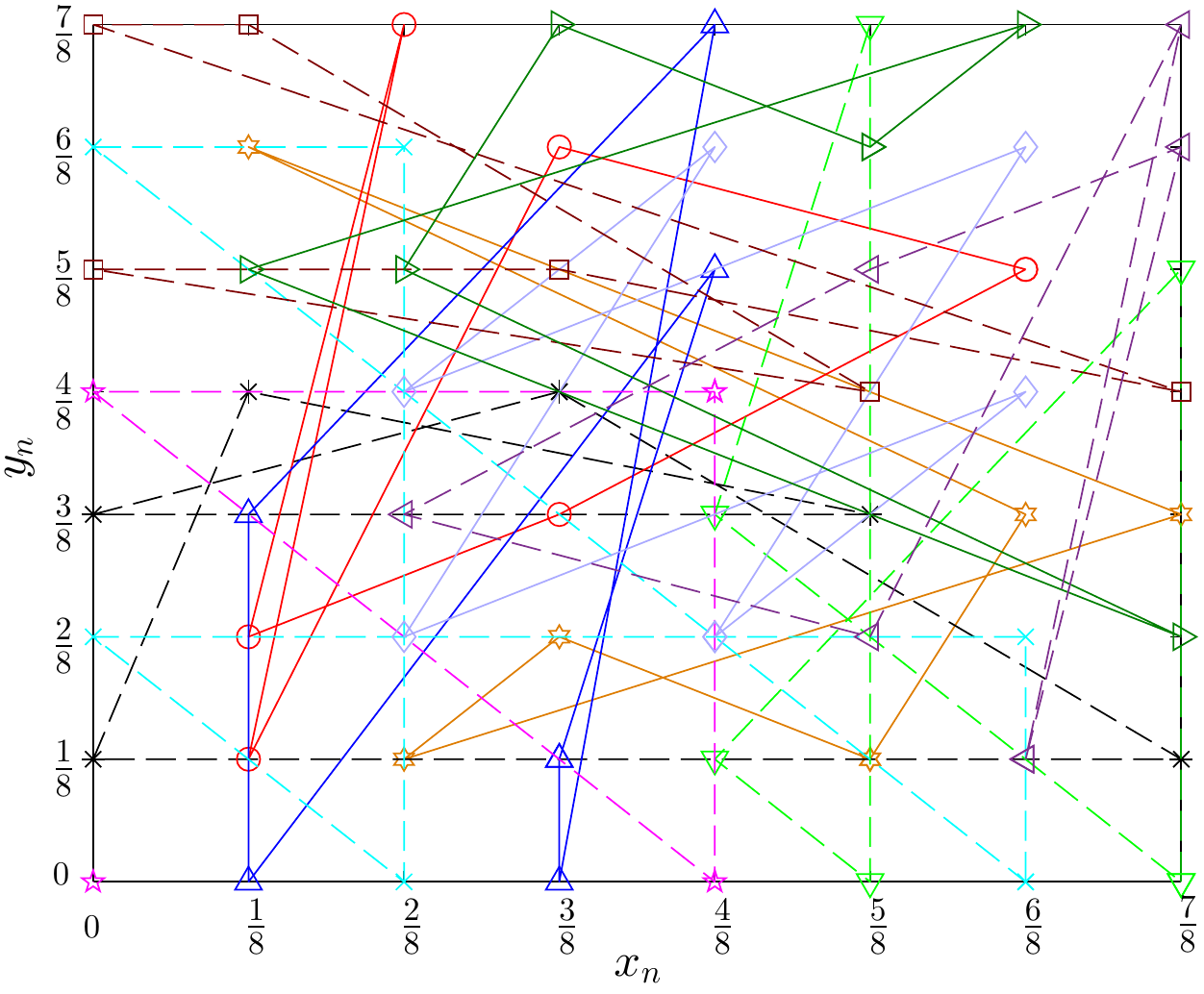}		
	\end{minipage}
\caption{Portrait of Cat map~(\ref{eq:ArnoldInteger}) with $N={2^3}$ where $(p, q)=(1, 3)$.}
\label{fig:upo}
\end{figure}

But, knowledge of the cycle structure of Cat map can be used in the following aspects:
\begin{itemize}
\item Disclosing some skeleton of Cat map in other domains. As shown in Fig.~\ref{fig:SMNcat}, when $p$ and $q$ are both odd, there is a cycle of period 3 in SMN of
Cat map implemented in Galois ring $\mathbb{Z}_{2^{\hat{e}}}$. This agree with the classic statement ``Period three implies Chaos" given in \cite{lity:3:AMM75}.

\item Discarding the initial conditions corresponding to very short period.
This problem is vitally important in real applications in \cite{Falcioni:PRNS:PRE2005,Curiac:path:DSJ2015}.

\item Avoiding the points resulting in collision. As shown in Fig.~\ref{fig:SMNcat}, different points in the same cycle may evolve into the same point, which may result in
collision for the hashing scheme proposed in \cite{Kanso:hash:ND2015}.

\item Severing as a prototype for analyzing dynamics degradation in chaotic maps implemented in digital computer discussed in \cite{Boghosian:Pathology:ATS19}.


\end{itemize}

\section{Conclusion}

This paper analyzed the structure of the 2-D generalized discrete Arnold's Cat map by its functional graph.
The explicit formulation of any iteration of the map was derived. Then, the precise cycle distribution of the generlized discrete Cat map in a fixed-point arithmetic domain was derived perfectly. The seriously regular patterns of the phase space of Cat map implemented in digital computer were reported to dramatically different from that in the infinite-precision torus. There exists non-negligible number of short cycles no matter what the period of the whole Cat map is. The analysis method can be extended to higher-dimensional Cat map and other iterative chaotic maps.

\bibliographystyle{IEEEtran_doi}
\bibliography{CAT}

\renewenvironment{IEEEbiography}[1] {\IEEEbiographynophoto{#1}}  {\endIEEEbiographynophoto}

\end{document}